\documentclass[11pt]{article}
\usepackage[utf8]{inputenc}
\usepackage[nottoc,notlot,notlof]{tocbibind}
\usepackage[colorlinks=false, linktocpage=true]{hyperref}
\usepackage{authblk}
\usepackage{graphicx}
\usepackage{fullpage}
\usepackage{amsthm,amsmath,mathscinet}
\usepackage[usenames, dvipsnames]{color}
\usepackage[none]{hyphenat}
\usepackage{cite}

\usepackage{thm-restate}
    \newtheorem{theorem}{Theorem}
    
    \newtheorem{lemma}[theorem]{Lemma}

\usepackage{url}
\usepackage[capitalise]{cleveref}

\def\S{\mathcal{S}}

\def\bigO{\mathcal{O}}

\def\alld{\texttt{ALLD}}
\def\allc{\texttt{ALLC}}
\def\tft{\texttt{TFT}}
\def\atft{\texttt{ATFT}}
\def\gtft{\texttt{GTFT}}
\def\wsls{\texttt{WSLS}}

\def\crate{C}

\def\d{\ \textrm{d}}
\def\IN{{\operatorname{IN}}}
\def\pfix{\rho^N}
\def\pim{\pi^M}
\def\pir{\pi^R}
\def\freq{f}
\def\freqs{\textbf{f}}

\newcommand{\figintro}{{\bf Fig.~1}}
\newcommand{\figMeanCooperationPartOne}{{\bf Fig.~2}}
\newcommand{\figMeanCooperationPartTwo}{{\bf Fig.~3}}
\newcommand{\figAbundanceReactiveStrategies}{{\bf Fig.~4}}
\newcommand{\figThreeStrategiesHeat}{{\bf Fig.~5}}
\newcommand{\figMemoryOne}{{\bf Fig.~S1}}
\newcommand{\figTwoStrategies}{{\bf Fig.~S2}}
\newcommand{\figThreeStrategiesHump}{{\bf Fig.~S3}}
\newcommand{\figAlldAlldGtft}{{\bf Fig.~S4}}
\newcommand{\figPhenotypic}{{\bf Fig.~S8}}
\newcommand{\figMemoryTwo}{{\bf Fig.~S9}}

\newcommand{\thmcrate}{{\bf Lemma 1}}
\newcommand{\thmuone}{{\bf Theorem 1}}
\newcommand{\thmuonecont}{{\bf Lemma 2}}
\newcommand{\thmunillstwo}{{\bf Theorem 2}}
\newcommand{\thmunillsthree}{{\bf Theorem 3}}
\newcommand{\thmfpin}{{\bf Lemma 3}}

\def\sne#1{\medskip\noindent\textbf{#1.}\ }

\newcommand{\hide}[1]{{\color{pink} }}

\def\SIOne{{\bf Appendix~1}}
\def\SITwo{{\bf Appendix~2}}
\def\SIThree{{\bf Appendix~3}}
\def\SIFour{{\bf Appendix~4}}
\def\ie{that is, }

\newcommand{\minimatrix}[2]{{\footnotesize \!\!
\begin{array}{c}
#1\\[-0.04cm]
#2
\end{array}}\!}
\newcommand{\minivector}[3]{#1_{\tiny \!\!\!\!
\begin{array}{c}
#2\\[-0.04cm]
#3
\end{array}}\!\!}

\title{Mutation enhances cooperation in direct reciprocity}

\author[a]{Josef Tkadlec}
\author[b,2]{Christian Hilbe} 
\author[a,c,1,2]{Martin A. Nowak}

\affil[a]{Department of Mathematics, Harvard University, Cambridge, MA 02138, USA}
\affil[b]{Max Planck Research Group `Dynamics of Social Behavior', Max Planck Institute for Evolutionary Biology, Plön, Germany}
\affil[c]{Department of Organismic and Evolutionary Biology, Harvard University, Cambridge MA~02138,~USA}
\affil[2]{These authors contributed equally.}
\date{}

\begin{document}

\maketitle

\begin{abstract}
Direct reciprocity is a powerful mechanism for evolution of cooperation based on repeated interactions between the same individuals. But high levels of cooperation evolve only if the benefit-to-cost ratio exceeds a certain threshold that depends on memory length. For the best-explored case of one-round memory, that threshold is two. Here we report that intermediate mutation rates lead to high levels of cooperation, even if the benefit-to-cost ratio is only marginally above one, and
even if individuals only use a minimum of past information. This surprising observation is caused by two effects.  First, mutation generates diversity which undermines the evolutionary stability of defectors. Second, mutation leads to diverse communities of cooperators that are more resilient than homogeneous ones. This finding is relevant because many real world opportunities for cooperation have small benefit-to-cost ratios, which are between one and two, and we describe how direct reciprocity can attain cooperation in such settings. Our result can be interpreted as showing that diversity, rather than uniformity, promotes evolution of cooperation.
\end{abstract}

In evolutionary game theory, cooperation is an action in which an individual voluntarily incurs a cost to give a benefit to someone else. 
While socially beneficial, cooperation is opposed by natural selection unless a mechanism for evolution of cooperation is in place~\cite{nowak:Science:2006,skyrms2014evolution}.
One such mechanism is direct reciprocity:
when the same two individuals interact repeatedly, mutual cooperation becomes a viable option~\cite{trivers:QRB:1971,axelrod:book:1984,boyd:Nature:1987,boyd:JTB:1989,Kraines:TaD:1989,nowak:Nature:1993,killingback:PRSB:1999,killingback:AmNat:2002,fischer2013fusing,garcia:jet:2016,akin:JDG:2017,Hilbe:NHB:2018,Glynatsi:SciRep:2020,Glynatsi:HSSC:2021}.

The phenomenon of cooperation can be described quantitatively by the donation game~\cite{sigmund:book:2010}, which is a simplified prisoner's dilemma~\cite{rapoport:book:1965}. 
In the donation game, each of two players chooses between cooperation and defection (\figintro\textbf{a}).
Cooperation means paying a cost, $c>0$, for the other player to receive a benefit, $b> c$.
Defection incurs no cost and causes no benefit.
When players only interact for one round of the donation game, evolutionary game theory predicts that both players learn to defect~\cite{nowak:book:2006}. 
This prediction, however, changes when the game is repeated. 
In that case, individuals can react to their co-player's previous behavior. 
They can employ conditionally cooperative strategies, such as Grim-Trigger~\cite{sigmund:book:2010}, Tit-for-Tat~\cite{rapoport:book:1965}, Generous Tit-for-Tat~\cite{molander:jcr:1985,nowak:Nature:1992a} or Win-Stay Lose-Shift~\cite{Kraines:TaD:1989,nowak:Nature:1993} to incentivize their co-player to cooperate. With such conditional strategies, mutual cooperation can be sustained as a Nash equilibrium. 


Full cooperation, however, is not the only possible equilibrium outcome of the repeated donation game.  On the contrary, the so-called Folk theorem guarantees the existence of a multitude of equilibria with all possible levels of cooperation, provided that each player gets at least the payoff for mutual defection~\cite{friedman:RES:1971}.
For example, in addition to cooperating in every round, there are equilibria in which players defect unconditionally, or in which they alternate between cooperation and defection~\cite{stewart:pnas:2014}.  


Because there are many equilibria, it becomes natural to ask which equilibrium emerges in populations of evolving players. This question can be explored with computer simulations of stochastic evolutionary dynamics~\cite{nowak:Nature:2004,Garcia:FRAI:2018,Hindersin:SciRep:2019}. 
In these simulations, players can choose among many different strategies for the repeated interaction. Over time, they abandon strategies that yield inferior payoffs, and instead adopt strategies that perform comparably well.  By analyzing the resulting evolutionary trajectories, researchers explore how likely players learn to cooperate, and which strategies they eventually use. 


The results of these individual-based simulations depend on a number of parameters, which include 
the benefit-to-cost ratio,  the population's size, the intensity of selection, and the mutation rate. The latter specifies how often players randomly explore new strategies. 
The values of these parameters not only affect whether or not cooperation evolves, but 
also how long it takes for populations to converge, and whether or not analytical approximations are feasible. 
One way to minimize computation time is to assume that mutations are exceedingly rare~\cite{fudenberg2006imitation,wu:JMB:2012,mcavoy:jet:2015}. 
In this limit, populations are homogeneous most of the time.
Only occasionally a mutant strategy arises, and this mutant either fixes in the population or goes extinct before the next mutation occurs. 
Since there is an explicit formula for the mutant's fixation probability~\cite{nowak:Nature:2004}, simulations that make use of the rare-mutation assumption tend to be many orders of magnitudes faster than conventional simulations~\cite{Hindersin:SciRep:2019}. 
In addition, simulation results can be interpreted more easily when mutations are rare, because the evolving population compositions are often closely connected to the Nash equilibria of the game~\cite{stewart:pnas:2014}. 
Due to these advantages, the rare-mutations assumption has become a standard approach to explore the evolution of direct reciprocity~\cite{nowak:AAM:1990,stewart:games:2015,Reiter:ncomms:2018,hilbe:PNAS:2013,stewart:pnas:2013,stewart:pnas:2016,Donahue:NComms:2020,Schmid:NHB:2021,park:NComms:2022,kurokawa:PRSB:2009,van-segbroeck:prl:2012,pinheiro:PLoSCB:2014}. 

However, there is by now substantial evidence from experimental games among human subjects suggesting that empirical mutation rates are sizeable~\cite{traulsen:PNAS:2010}. In particular, estimated mutation rates are far beyond the threshold  for which the rare-mutation approximation is valid~\cite{grujic:SciRep:2014}. This observation raises the question how individuals learn to engage in direct reciprocity when mutations occur more frequently. This is the question that we explore in the present paper.


If mutations are frequent, evolutionary dynamics lead to communities in which many different strategies co-exist. 
Surprisingly, we find that such diverse communities facilitate cooperation.
The previous literature based on rare mutations emphasized that cooperation can only evolve if the benefit-to-cost ratio, $b/c$, is large. 
For example, the well-known strategy Win-Stay Lose-Shift can only maintain cooperation if $b/c>2$. 
 Among reactive strategies (which only take into account the co-player's previous action), it takes an even larger benefit-to-cost ratio for  cooperation to evolve~\cite{Baek:SciRep:2016}. 
Here, we show that once mutation rates are non-negligible, high cooperation levels can occur even as the benefit-to-cost ratio approaches one, the theoretical minimum.
To this end, we present extensive simulations for various strategy sets, including the set of stochastic reactive strategies, stochastic memory-1 strategies, and deterministic memory-2 strategies. 
In each case, we find that intermediate mutation rates are able to facilitate cooperation in parameter regions in which cooperative populations are unlikely to emerge otherwise. 
The emerging diverse communities destabilize the equilibrium around all out defectors, while keeping the cooperative equilibria stable.
Thus, certain levels of diversity -- rather than uniformity -- promote cooperation.

\section*{Results}

\noindent 
{\bf Model framework.}
We consider a population of size~$N$. Individuals engage
in repeated donation games~(\figintro). 
For most of the main text, we assume that individuals make their decision whether or not to cooperate based on reactive strategies~\cite{nowak:Nature:1992a}. 
Such strategies consist of two parameters: $p$ is the probability to cooperate if the co-player has cooperated in the previous round,
whereas $q$ is the probability to cooperate if the co-player has defected in the previous round.
A strategy is called deterministic if all cooperation probabilities are either zero or one. 
If at least one probability is in between, the strategy is called stochastic. 
The space of stochastic reactive strategies is the unit square, and the deterministic strategies correspond to the four corners of this square. 
Reactive strategies include:
always defect $\alld=(0,0)$,
always cooperate $\allc=(1,1)$,
the random strategy $(0.5,0.5)$,
tit-for-tat $\tft=(1,0)$, and
generous tit-for-tat $\gtft=(1,q)$, where $q\in [0,1]$ is the probability of forgiveness or the level of generosity.
We note that it is straightforward to compute the cooperation rate $\crate(S,S')$ for a reactive strategy $S$ when facing a reactive strategy $S'$. This in turn also allows us to compute the resulting payoff $\pi(S,S')$, \ie the average benefit minus cost derived from that interaction~\cite{press2012iterated}. For details, see~\SITwo.
We focus on reactive strategies for simplicity; analogous results hold when players use memory-1~(\figMemoryOne) or memory-2 strategies~(\figMemoryTwo). 

\begin{figure}[h] 
  \centering
   \includegraphics[width=0.9\linewidth]{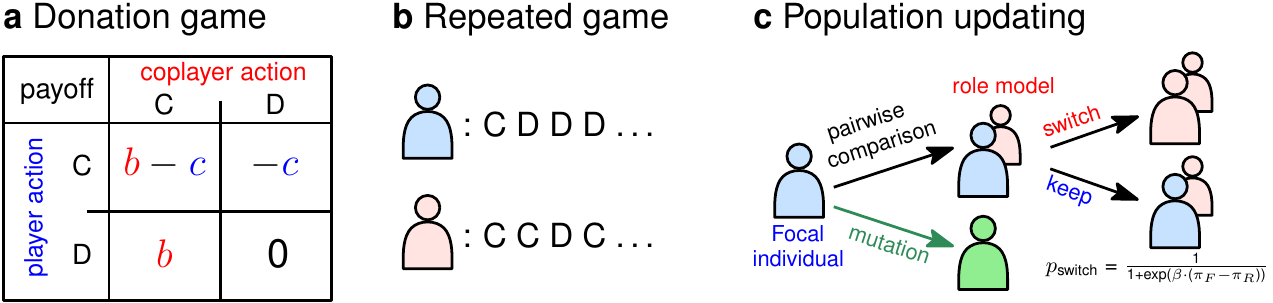}
\caption{
\textbf{Evolutionary dynamics of cooperation.} 
\textbf{a,} In the donation game, players choose between cooperate and defection. Cooperation incurs a cost $c$ and provides a benefit $b$ to the co-player. Defection incurs no  cost and provides no benefit.
\textbf{b,} We consider the repeated donation game, in which players can use conditional strategies that depend on the outcome of previous rounds. 
\textbf{c,} In each evolutionary update step, a focal individual, $F$, either explores a new random strategy (with probability $u$), or compares his own payoff to that of a random role model, $R$. He is then more likely to switch to the role model's strategy if she performs better than him.
}
\label{fig:model-game}
\end{figure}

At any point in time, the composition of the population is described by a list of the employed strategies $(S_1, S_2, ..., S_N)$.
Each individual $i$ derives an average payoff
$ \pi_i=\sum_{j\ne i} \pi(S_i,S_j) / (N-1)$ from all pairwise interactions. 
Individuals learn to adopt more profitable strategies over time. 
To describe the resulting dynamics, several processes have been proposed, including processes based on stochastic best responses~\cite{blume:GEB:1993} or based on stochastic imitation~\cite{szabo:PRE:1998}.
Here, we use a pairwise comparison process~\cite{traulsen:PRE:2006b}, which is a variant of stochastic imitation. 
When a focal individual $F$ revises their strategy, they can do so in two ways (\figintro{\bf c}).
With probability $u$, the focal individual adopts a new strategy at random, which represents mutation. 
With probability $1-u$, the focal individual considers imitating the strategy of another population member, which corresponds to reproduction and selection. 
To this end, the focal player randomly picks a role model $R$ from the population and compares its payoff, $\pi_F$, to that of the role model, $\pi_R$. 
The focal player switches to the role model's strategy with a probability given by the Fermi function, $1/\big[1+\exp(\beta \cdot\! \big(\pi_F-\pi_R)\big)\big]$; otherwise the focal individual keeps their strategy. 
The parameter $\beta \ge0$ measures the intensity of selection. 
It reflects how clearly payoffs can be evaluated. 
If $\beta=0$, the imitation probability simplifies to one half, meaning that payoffs become irrelevant and we are in the realm of neutral evolution. 
In the other limit $\beta\rightarrow\infty$, the focal player adopts the role model's strategy only if $\pi_R \ge \pi_F$, and selection becomes very strong. 
For most our simulations we choose $\beta =10$, which represents an intermediate intensity of selection.

We consider two versions of this process. 
For our simulations, we consider a Wright-Fisher type model with non-overlapping generations~\cite{imhof:JMB:2006}: 
in each time step, all players are given the opportunity to revise their strategy. 
For analytical calculations, we complement this framework with a Moran type model~\cite{nowak:Nature:2004}, in which in each step only one randomly chosen player can revise their strategy. 
In both cases we obtain similar results, but the model with non-overlapping generations is computationally more efficient. 
Our simulations start out with random populations. 
Over time, players adopt new strategies according to the above process. 
The results depend on the benefit-to-cost ratio~$b/c$, the mutation rate~$u$, the population size~$N$, and the intensity of selection~$\beta$.
In the following, we explore how these parameters affect evolutionary outcomes. 
The quantity of interest is the average cooperation rate~$\crate$, with the average being taken over all pairwise interactions over sufficiently long time.
\\


\noindent
{\bf The diversity effect.}
In \figMeanCooperationPartOne{} we show how the average cooperation rate~$\crate$ depends on the benefit-to-cost ratio~$b/c$. In the limit of rare mutations,  $u\to 0$, the process is relatively well understood~\cite{imhof:PRSB:2010}. In that setting the cooperation rate increases only very slowly with rising benefit-to-cost ratio. But if we add mutation we find a dramatic increase in the cooperation rate. For population size $N=100$, we find that mutation rates between $u=0.01$ and $u=0.03$ are ideal for reactive strategies, while slightly higher mutations rates (from $u=0.03$ to $u=0.05$) are ideal for memory-1 strategies. The optimum mutation rate depends on the exact value of $b/c$.

\begin{figure*}[t] 
  \centering
   \includegraphics[width=0.3\linewidth]{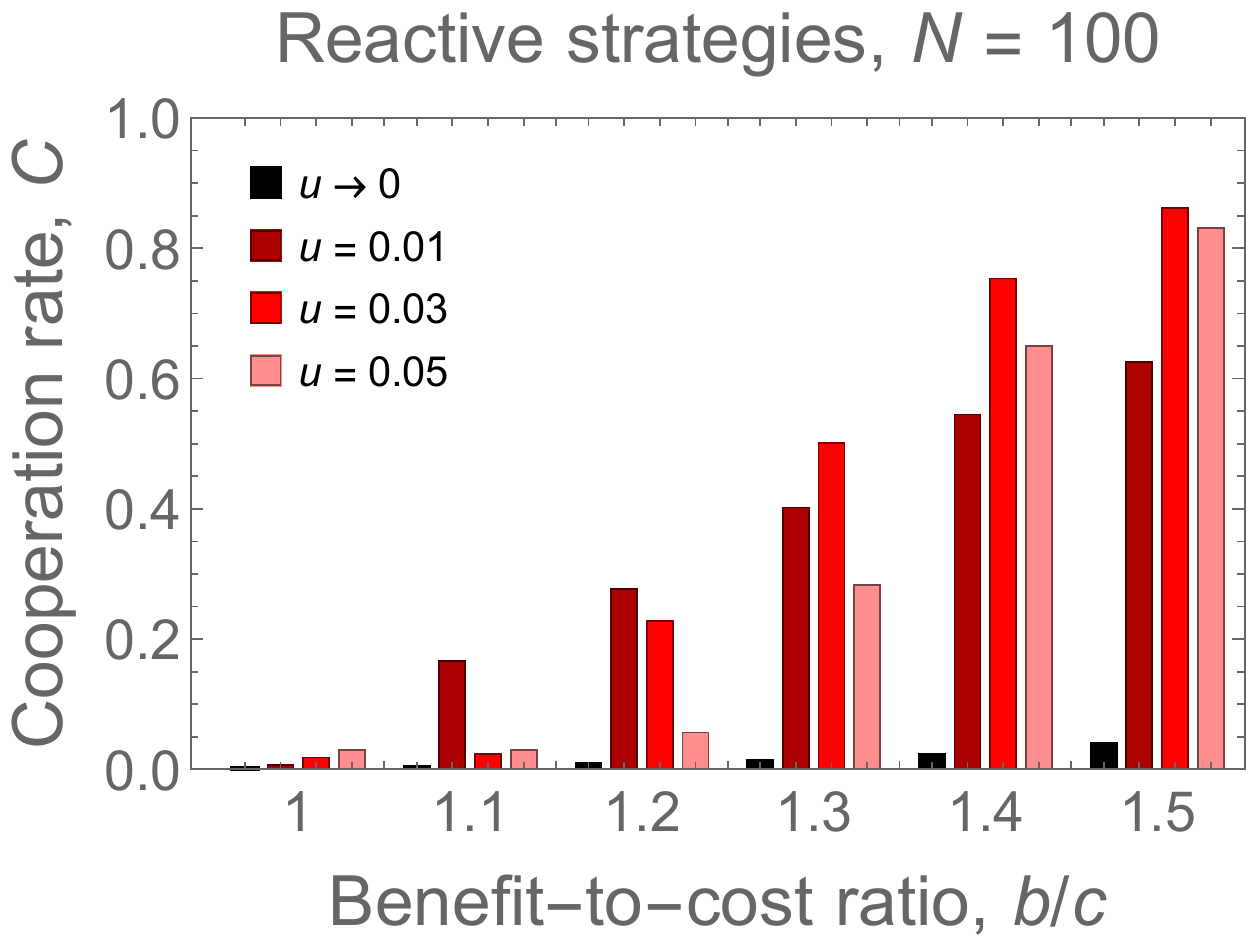}%
   \includegraphics[width=0.3\linewidth]{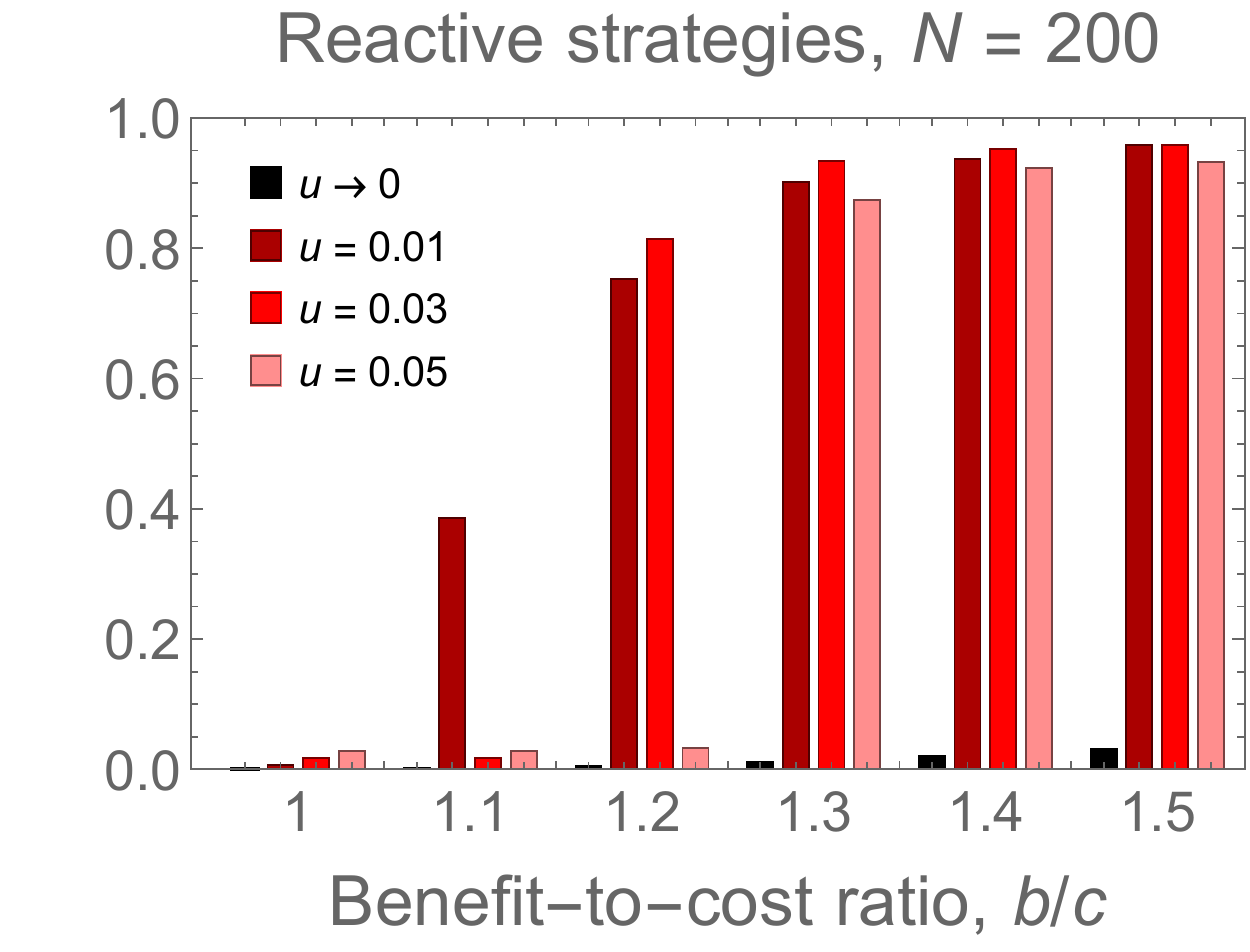}%
   \includegraphics[width=0.3\linewidth]{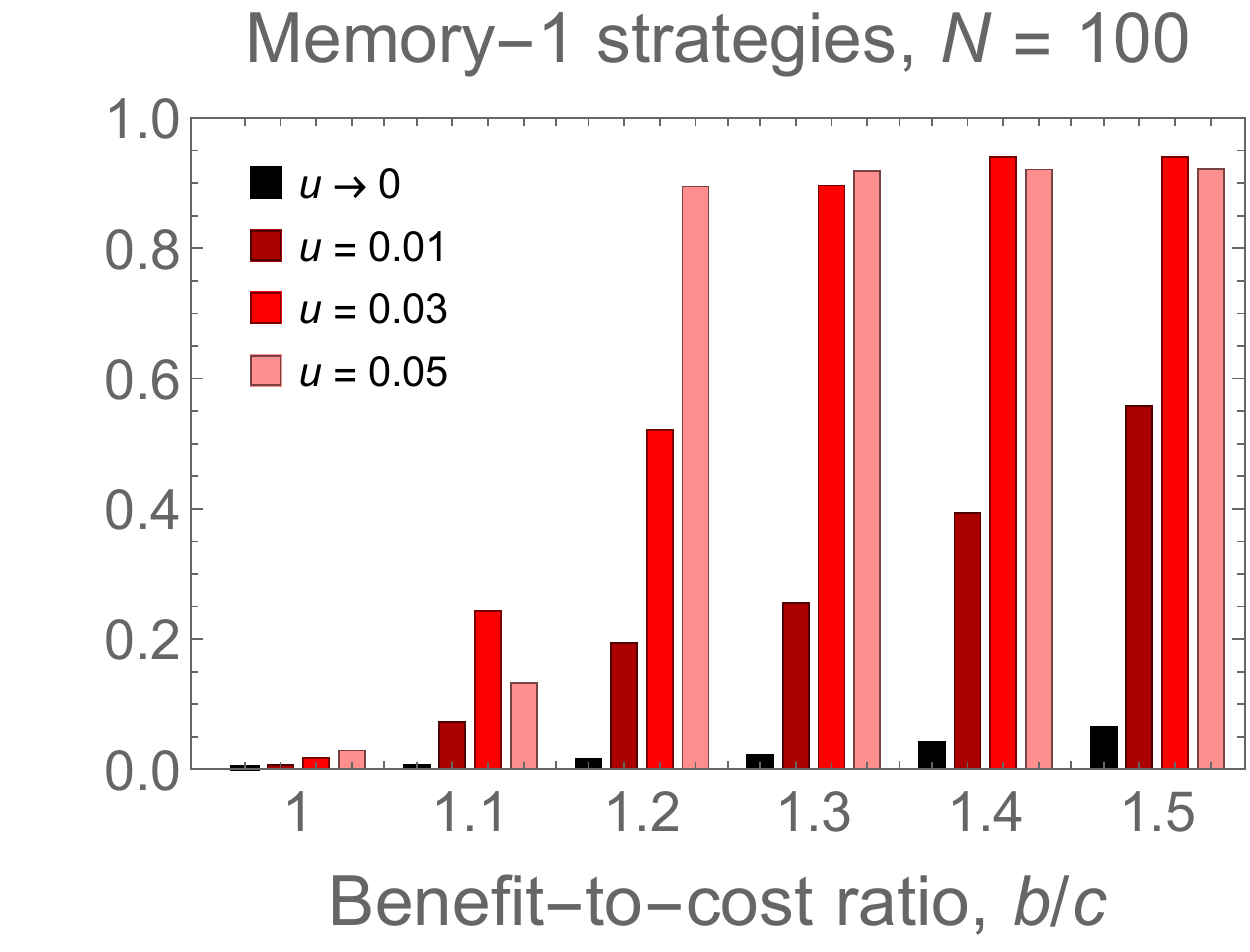}
\caption{
\textbf{The effect of diversity on cooperation.} In the limit  of rare mutations ($u\to 0$, black bars) the average cooperation rate increases very slowly with the benefit-to-cost ratio. Adding mutation (red bars) substantially enhances cooperation for small benefit-to-cost ratios. 
Parameters: $\beta=10$; simulations are run for at least $10^9$ updates to get reliable averages.
}
\label{fig:f0-c-plot}
\end{figure*}

In the limit of rare mutation, $u\to 0$, the population is mostly homogeneous at any one time and is exploring the strategy space by making transitions between states that are dominated by single strategies. For larger mutation rates, communities are more diverse and, unexpectedly, this enhances the average cooperation rate.  Our goal is to understand this surprising effect. Why does diversity promote cooperation? In the limit of vanishing mutation, all players update their strategies based on performance, but in the presence of mutation, some players choose randomly. Why do random choices augment cooperation?\\

\noindent
{\bf Characteristic curve and optimum mutation rate.}
In the donation game, the benefit-to-cost ratio varies between one and infinity.
Consequently, the cost-to-benefit ratio varies between zero and one.
For a given mutation rate $u$ and a population size $N$, we can plot the average cooperation rate in the population versus the cost-to-benefit ratio, $c/b$.
We call this graph the characteristic curve of the evolutionary process (\figMeanCooperationPartTwo\textbf{a}).
Since increasing the cost-to-benefit ratio makes cooperation less rewarding, all characteristic curves are expected to decline monotonically. For $u\to 0$, we find very low levels of cooperation if $c/b>1/2$. For $u=0.01$, high levels of cooperation are observed for population sizes $N=100$ and 200 even if  $c/b>1/2$.


\begin{figure}[h!] 
  \centering
\includegraphics[width=0.9\linewidth]{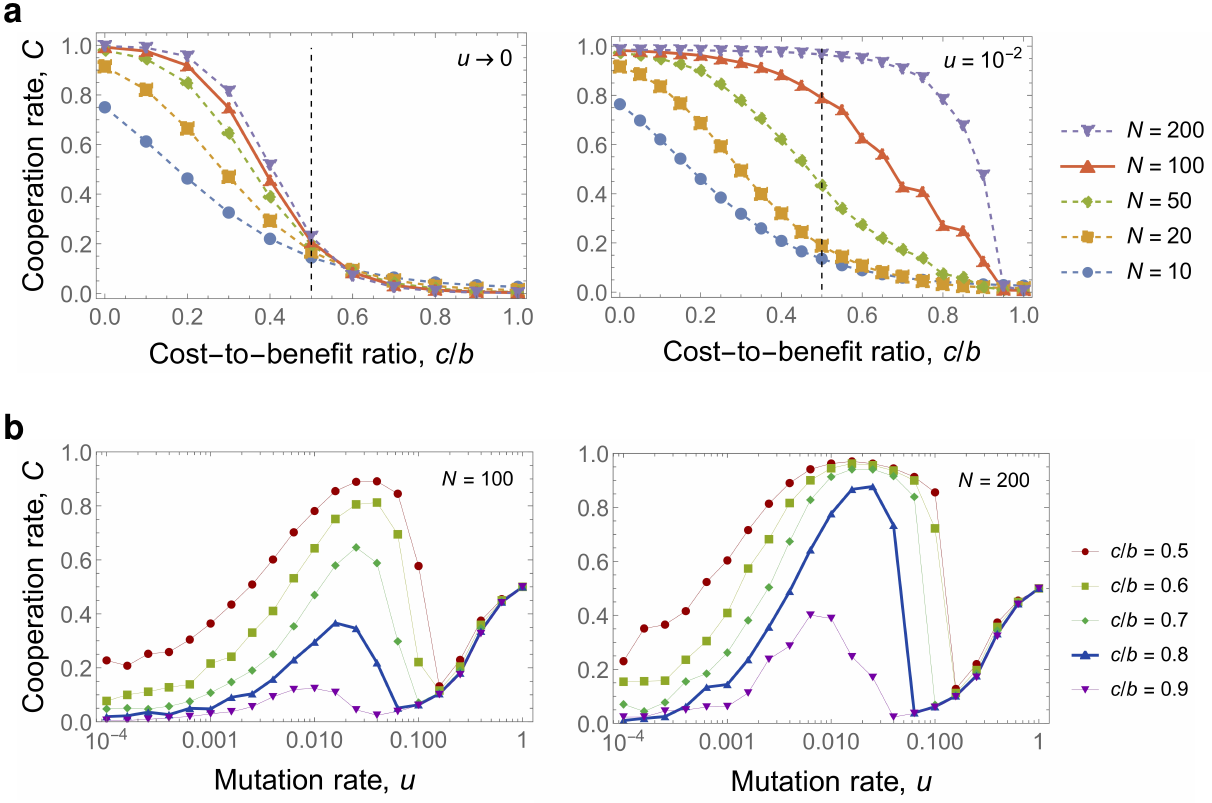}
\caption{
\textbf{Characteristic curve and the optimum diversity for reactive strategies.}
\textbf{a,} The characteristic curve of an evolutionary process of cooperation and defection is the graph of average cooperation rate versus cost-to-benefit ratio. In the limit of rare mutations ($u\to 0$, left), no substantial cooperation occurs for $c/b>0.5$. 
For mutation rate $u=10^{-2}$ (right), we observe substantial levels of cooperation even if $c/b>0.5$, especially when $N\ge 50$.
\textbf{b,} The average cooperation rate~$C$ as a function of the mutation rate $u$, for $N=100$ (left) and $N=200$ (right).
When there are only mutations, $u\!=\!1$, the cooperation rate approaches $C\!=\!0.5$, regardless of the cost-to-benefit ratio $c/b$ (see~\thmuonecont{} in~\SITwo).
As $u$ goes down, the cooperation rate first drops toward zero,
then jumps up toward one (here at around $u\approx 10^{-1}$),
and finally it drops towards zero again (here at around $u\approx 10^{-4}$).
We call these phenomena the ``valley'' and the ``hump''. The optimum mutation rate is around $u\approx 10^{-2}$. We use $\beta=10$. 
}
\label{fig:cont-hump}
\end{figure}

To examine the role of diversity in more detail, another perspective is useful.
In~\figMeanCooperationPartTwo\textbf{b} we show the average cooperation rate~$\crate$ as a function of the mutation rate~$u$. Proceeding from high ($u=1$) to low mutation rates $(u\approx10^{-4})$, we observe three trends: First, the cooperation rate declines from $0.5$ towards near zero; subsequently it rises suddenly to very high values (near one); and then it declines again. 
We refer to those trends as the \textit{valley} and the \textit{hump}.
The exact locations of the minimum of the valley and the maximum of the hump depend on the cost-to-benefit ratio and the population size.
However, both the valley and the hump occur consistently across many combinations of parameters considered.\\

\noindent
{\bf Stationary distribution.}
To understand the two effects of valley and hump, we show in~\figAbundanceReactiveStrategies{} how often each strategy in the space of all reactive strategies is used at various mutation rates.
The heat maps suggest that two regions of the strategy space are visited predominantly. 
The first region corresponds to a set of defective strategies $(p,0)$, with $0\le p \le c/b$, including the strategy $\alld=(0,0)$. 
The second region consists of generous tit-for-tat strategies $\gtft=(1,q)$ with $q$ satisfying $0 \le q \le 1-c/b$.
Both regions are favored by selection for all mutation rates: they are  visited more often than expected under neutrality. 
However, the relative abundance of each region changes with $u$. 
For small mutation rates such as $u=10^{-4}$, the region of defectors is most abundant. 
For larger mutation rates up to $u=10^{-1}$, individuals predominantly use $\gtft$.\\

\begin{figure}[h] 
  \centering
   \includegraphics[width=0.9\linewidth]{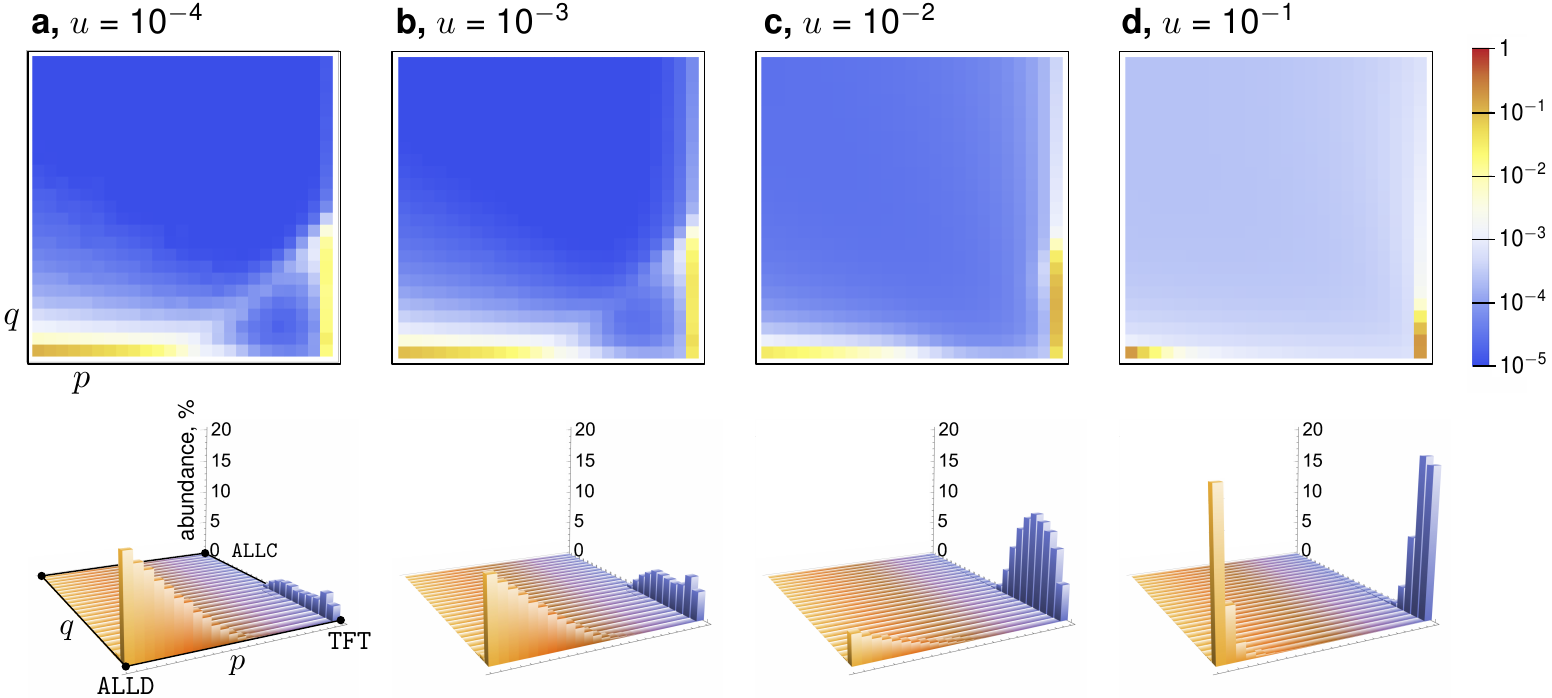}
\caption{
\textbf{Frequencies of reactive strategies.}
\textbf{a-d:} For 4 distinct values of the mutation rate, $u =10^{-4},10^{-3},10^{-2},10^{-1}$, we
simulate evolutionary dynamics for at least $10^9$ steps. We collect the 
appearing strategies into a $25\times 25$ grid and
plot their relative abundance as a heat map (top row, log scale) and as a bar chart (bottom row, linear scale).
In all 4 cases, the strategies with non-negligible frequency have either $q\approx0$ (orange peak, bottom left) or $p\approx1$ (blue peak, right).
For intermediate mutation rate, $u=10^{-2}$ (third column), the blue peak has more mass than the orange one, and we observe high overall cooperation rates. Parameters: $N=100$, $c/b=0.5$, $\beta=10$.
}
\label{fig:cont-heatmaps}
\end{figure}

\noindent
{\bf Detailed analysis of a reduced strategy set.} 
To gain intuition for these findings, it is useful to consider a reduced strategy set. 
This set ought to be large enough to reproduce the above findings, yet small enough to be tractable in detail. 
In \SITwo{} we show that a set with two strategies is not sufficient to reproduce all qualitative findings. Specifically, if individuals can only choose between $\alld\!=\!(0,0)$ and $\gtft=(1,q)$, we obtain neither the valley nor the hump (\figTwoStrategies).  Intuitively, this two-strategy system fails to capture one key feature of the full system, namely that all strategies $(1,q)$ with $q\in[0,1]$ are neutral with respect to each other. This neutrality implies that the mass of the blue peak in~\figAbundanceReactiveStrategies{} can move freely along the edge $p=1$ of the unit square $[0,1]^2$.

To capture this effect, we study a three-strategy system consisting of $\alld$, $\gtft$ and $\allc$.
For this system, we observe both the valley and the hump, independent of whether mutations introduce all three strategies equally often or whether they are biased towards $\alld$, see~\figThreeStrategiesHump.

\begin{figure}[h]
  \centering
   \includegraphics[width=0.9\linewidth]{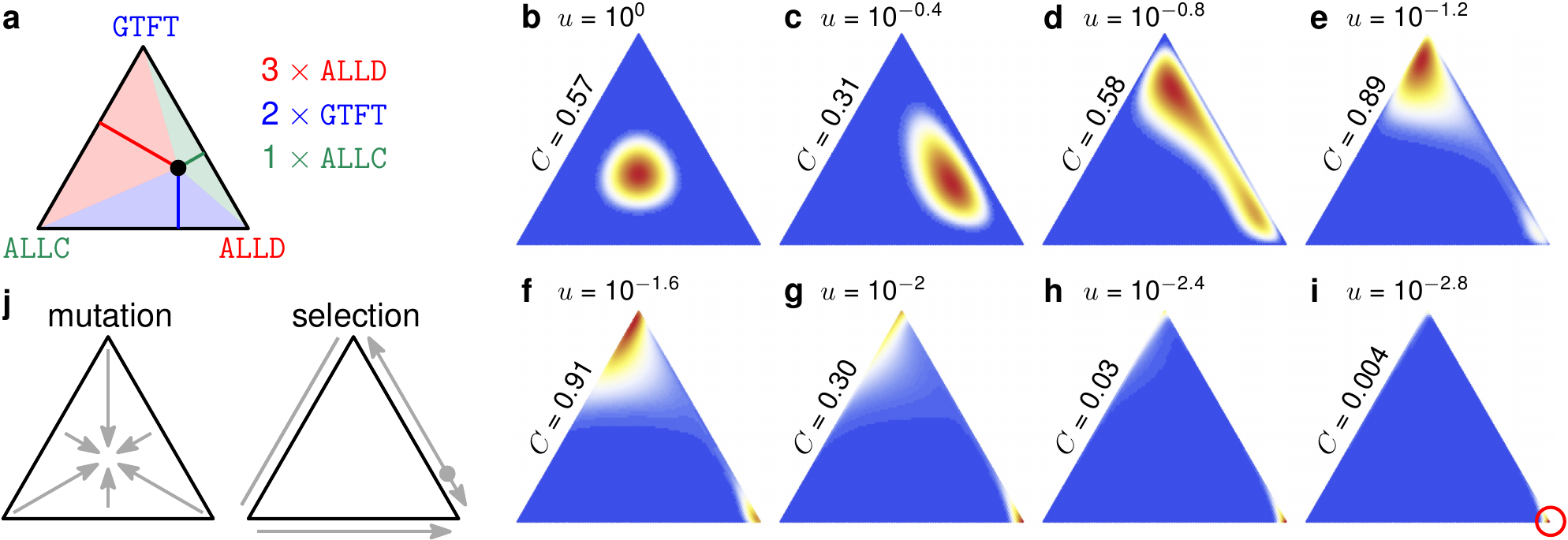}
\caption{
\textbf{Evolution of the reduced three strategy system $\S_3=\{\alld,\gtft,\allc\}$.}
\textbf{a,} Each point in the simplex represents a composition of the population.
\textbf{b-i,} Relative frequencies of the population composition (red is high).
as $u$ decreases from \textbf{b,} $u=10^0=1$ to \textbf{i,} $u=10^{-2.8}$. 
\textbf{b,} For $u=1$ the population typically consists of roughly $1/3$ of each of $\alld$, $\gtft$, $\allc$ and the average cooperation rate is $5.1/9\doteq0.57$ (see~\thmuone{} in \SITwo).
\textbf{e-f,} For $u\in[10^{-1.2},10^{-1.6}]$ the population typically consists mostly of $\gtft$ players and the cooperation rate is $\approx 0.9$.
\textbf{h-i}, For $u\le 10^{-2.4}$ virtually all players play $\alld$ virtually all the time.
\textbf{j,} The effects of mutation and selection on the total mass.
Parameters: $N=100$, $c=0.7$, $\beta=10$. For $\gtft$ we use $(1,0.1)$.
}
\label{fig:3-temp}
\end{figure}

To explain the valley and the hump in this reduced strategy space, we plot the relative abundance of each possible population composition. Any compositions can be depicted as a point in a simplex with corners $\alld$, $\gtft$, $\allc$.
The relative frequency of each strategy is proportional to the distance of the point from the opposite side (\figThreeStrategiesHeat\textbf{a}).
Corners correspond to homogeneous populations, where all individuals play the same strategy.

The evolutionary dynamics on the simplex can be interpreted most easily in the limiting cases, when there are either only mutations ($u=1$) or very rare mutations ($u\ll1$). 
When there are only mutations, individuals use each of the three available strategies with equal probability. 
Since there is no selection, the law of large numbers implies that populations are concentrated around the center of the simplex (\figThreeStrategiesHeat\textbf{b}). 
In such populations, the strategy receiving the highest payoff is $\alld$. 
If $u$ is slightly decreased, $\alld$ is favored by selection and becomes most frequent. 
Since $\allc$ is exploited by $\alld$, the remaining population members are more likely to adopt $\gtft$ rather than $\allc$ (\figThreeStrategiesHeat\textbf{c}).

Looking at the other extreme, when $u$ is very small (\figThreeStrategiesHeat\textbf{g--i}), populations are typically homogeneous. That is, populations are at one of the threes corners of the state space most of the time; they are less often on one of the edges; and they are least often in the interior. 
In this limit, selection makes the population oscillate around the simplex as follows (\figThreeStrategiesHeat\textbf{j}).
If the current population predominantly uses $\alld$, then $\gtft$ can invade with a certain probability~\cite{nowak:Nature:2004}. Once this happens, evolution leads towards a pure $\gtft$ state.
Because $\alld$ players are now absent, the strategies $\allc$ and $\gtft$ are neutral with respect to each other.
Thus, a pure $\gtft$ population drifts toward a mix of $\gtft$ and $\allc$.
Once the proportion of $\allc$ is high, the population becomes susceptible to invasion by $\alld$.
This pattern is consistent with the behavior of the full process in the limit of rare mutations, $u\to 0$.
In that limit, the population resolves to a pure state between every two consecutive mutations. 
As a consequence, the relative frequencies of the three pure states can be computed from the probabilities $p_{A\to B}$ that a single mutant with  strategy $B$ successfully invades and fixes in a resident population with strategy $A$, for $A,B\in\{\alld,\gtft,\allc\}$~\cite{fudenberg2006imitation}.
In particular, for large $N$ (and any $c$ and $q$), we show that the population spends virtually all the time at a pure $\alld$ state, see~\thmunillsthree{} in the \SITwo.

For intermediate mutation rates, all three strategies are typically present in the population simultaneously,
but selection events occur sufficiently often to drive inefficient strategies to low frequencies. 
As a result, we observe two peaks for intermediate $u$. 
One peak is near $\alld$ and the other one near $\gtft$ (\figThreeStrategiesHeat\textbf{d}), which is consistent with the behavior observed for the full strategy space (\figAbundanceReactiveStrategies). 
In this regime, mutations have two positive effects on the evolution of cooperation. 
On the one hand, they destabilize the peak around $\alld$: because mutations recurrently introduce $\gtft$ players, these $\gtft$ players are more likely to reach a critical number after which cooperation is more profitable. 
On the other hand, mutations make the peak around $\gtft$ more stable: because mutations recurrently introduce small minorities of defectors, $\gtft$ yields a strictly better payoff than $\allc$, which prevents $\allc$ from invading through neutral drift. 

This situation is reminiscent of a counterintuitive effect in rock-paper-scissors games. 
In these games, providing a payoff advantage to one strategy may eventually lower that strategy's frequency in equilibrium~\cite{tainaka:PLA:1993,frean:PRSB:2001}. 
This counterintuitive effect could explain our findings, as higher mutation rates seem to give a payoff advantage to defectors. 
However, we believe this explanation does not fully capture our results. 
First, in addition to (indirectly) affecting the payoffs of each strategy, mutations alter the evolutionary dynamics altogether. 
They provide the population with a larger pool of role models that can be imitated subsequently. 
Second, the dynamics between \alld, \gtft, and \allc{} does not follow a strict rock-scissors-paper cycle (see \figThreeStrategiesHeat\textbf{j}). 
Instead, the competition between \alld{} and \gtft{} is bistable, whereas the competition between \gtft{} and \allc{} is neutral. 
In each case, mutations favor cooperation: they help \gtft{} to overcome the invasion barrier (against \alld), and to resist neutral drift (against \allc). 

We emphasize that unconditional $\allc$ players catalyse the transition from the cooperative equilibrium to the defective one in the limit of rare mutation.
If we instead consider the triplet $\{\alld,\gtft,\alld\}$, where $\allc$ is replaced by another copy of  $\alld$, then high levels of cooperation occur for all low enough mutation rates (\figAlldAlldGtft).\\

\noindent
{\bf Beyond reactive strategies.}
To illustrate the robustness of our results, we also explored the space of stochastic memory-1 strategies~\cite{nowak:Nature:1993}.
These strategies are given by four parameters $p_{CC}$, $p_{CD}$, $p_{DC}$, $p_{DD}$ that denote the probabilities to cooperate after each of the possible outcomes of the last round, $CC$, $CD$, $DC$, $DD$, respectively. The space of memory-1 strategies is the hypercube, $[0,1]^4$. The space  includes all reactive strategies, but also Win-Stay Lose-Shift and many other strategies. In \figMemoryOne{} and in~\figMeanCooperationPartOne{} we show that our findings continue to hold in this broader space. In the limit of rare mutations, $u\to 0$, there is no substantial cooperation when the cost-to-benefit ratio exceeds 1/2. In contrast, for non-negligible mutation rates, cooperation does occur even if $c/b>1/2$. Moreover, the average cooperation rate as a function of $u$ again exhibits both the valley and the hump. 
In \SIFour{} we report similar results for (deterministic) memory-two strategies (\figMemoryTwo). 
Overall, these findings suggest that the positive effects of mutations are not restricted to one particular strategy space.

\section*{Discussion}

In repeated social dilemmas, individuals can sustain cooperation by reacting to their interaction partner's previous actions. 
This mechanism for cooperation is called direct reciprocity~\cite{trivers:QRB:1971,axelrod:book:1984,boyd:Nature:1987,boyd:JTB:1989,Kraines:TaD:1989,nowak:Nature:1993,killingback:PRSB:1999,killingback:AmNat:2002,fischer2013fusing,garcia:jet:2016,akin:JDG:2017,Hilbe:NHB:2018,Glynatsi:SciRep:2020,Glynatsi:HSSC:2021}.
For direct reciprocity, fully cooperative Nash equilibria exist if the benefit-to-cost ratio is greater than one, $b/c>1$~\cite{Akin:chapter:2016}. 
But there is still a range of other equilibria with lower levels of cooperation~\cite{stewart:pnas:2014}.
In particular, full defection remains an equilibrium for all benefit-to-cost ratios.
Therefore it is a question of evolutionary dynamics which equilibrium is chosen and how much cooperation is achieved on average.
When individuals react to their interaction partner's last decision, previous papers~\cite{hilbe:PNAS:2013,stewart:pnas:2013,Baek:SciRep:2016} suggest that high levels of cooperation only emerge if the benefit-to-cost ratio is greater than two, $b/c>2$.
Hence, the interval $2>b/c>1$ was void of cooperation.
Here we show how to attain high levels of cooperation for $b/c>1$ in general,
and thus for $2>b/c>1$ in particular.

Many previous studies of direct reciprocity explore selection dynamics either in the absence of mutation~\cite{nowak:AMC:1989,brandt:JTB:2006,hilbe:BMB:2011,grujic:jtb:2012,Rodriguez:JMB:2016,szabo:PRE:2000b,szolnoki:scirep:2014,szolnoki:pre:2014}, or in the limit of very low levels of mutation~\cite{nowak:AAM:1990,stewart:games:2015,Reiter:ncomms:2018,hilbe:PNAS:2013,stewart:pnas:2013,stewart:pnas:2016,Donahue:NComms:2020,Schmid:NHB:2021,park:NComms:2022,kurokawa:PRSB:2009,van-segbroeck:prl:2012,pinheiro:PLoSCB:2014} (for a more detailed discussion of the previous literature, see \SIOne).
In the context of evolutionary game theory, selection means to adopt strategies based on performance,
while mutation means to adopt strategies randomly.
We find that intermediate mutation rates allow high levels of cooperation even if benefit-to-cost ratios are only marginally above one.
Mutations generate diverse communities which we find to have two beneficial effects for evolution of cooperation.
First, diversity undermines defective equilibria by seeding clusters of potential invaders.
Second, diversity stabilizes cooperative equilibria by preventing neutral drift toward strategies that can be exploited.

The importance of mutant strategies has also been highlighted in earlier studies of direct reciprocity~\cite{boyd:Nature:1987,van-veelen:PNAS:2012,garcia:jet:2016}. 
These studies stress that any resident population becomes unstable when interacting with an appropriate ensemble of mutants.
Importantly, however, the reported effects of mutant strategies in these studies are symmetric. 
Just as there are mixtures of mutants that can invade an \alld{} equilibrium, there are mutants that can invade \wsls{} (or any other cooperative resident strategy).
In contrast, the effect of mutations that we report is asymmetric. 
We find that the variation introduced by mutations systematically favors cooperation. 
This effect is robust and it arises independently of the strategy spaces we considered. 
Interestingly, the positive effect of variation does not even require this variation to be heritable. 
Indeed, further simulations reported in \SIThree{} suggest that even phenotypic (non-heritable) variation in the player's behaviors can foster cooperation (\figPhenotypic). 

The essential question of evolution of cooperation is not only whether cooperation can evolve but for which benefit-to-cost ratio. 
The efficiency of a mechanism could be defined by the minimum benefit-to-cost ratio needed for the evolution of cooperation. 
We show that the efficiency of direct reciprocity is greatly enhanced by non-negligible mutation rates.  
Remarkably, the path to cooperation identified in this paper does not require complex strategies with extended memory capacities, as proposed earlier~\cite{hauert:PRSB:1997,hauert:JTB:2002b,stewart:scirep:2016,hilbe:pnas:2017,Li:NatCS:2022}. 
Instead, this path is already available to individuals with the most basic strategies of direct reciprocity. 
In our paper, cooperation does not evolve because individual strategies are sufficiently sophisticated. 
Instead cooperation evolves because the evolutionary process is sufficiently erratic. 
Our observation is both surprising and relevant.
It is surprising because one would not expect a-priori that replacing performance based update events with random ones could promote cooperation.
It is relevant because many opportunities for cooperation naturally arise in situations where the benefit-to-cost ratio is only slightly above one.

 
\subsection*{Methods}
Here we sketch the intuition behind the mathematical claims made in the main text.
The full mathematical proofs appear in the \SITwo.

\sne{Reactive strategies}
A reactive strategy is given by a pair $(p,q)$, where $p\in[0,1]$ (resp.\ $q\in[0,1]$) are the probabilities to cooperate in the next round, assuming that the co-player has just cooperated (resp.\ defected).
When two individuals employing reactive strategies $s$ and $s'$ play an infinitely repeated donation game, their payoffs $\pi(s,s')$, $\pi(s',s)$ derived from that interaction can be computed using a standard formula~\cite{press2012iterated}.
The same applies to the pairwise cooperation rates $\crate(s,s')$, $\crate(s',s)$, see~\thmcrate{} in the \SITwo.

\sne{Case $u=1$ (no selection)}
When $u=1$ then selection does not play any role and at each point in time, each individual plays a strategy selected uniformly at random from the space $\S$ of available strategies.
As a consequence, when the space $\S$ consists of finitely many strategies, by linearity of expectation the overall cooperation rate can be computed by averaging the pairwise cooperation rates among the strategies, see~\thmuone.
Moreover, when the space $\S$ includes all reactive strategies, we show that the average cooperation rate of the population is precisely $1/2$. This is because of a certain symmetry between cooperation and defection, see~\thmuonecont. 

\sne{Case $u\to 0$ (rare mutations)}
In the regime of rare mutations ($u\to 0$) we obtain a simplified process studied by Fudenberg and Imhof~\cite{fudenberg2006imitation}, as well as Imhof and Nowak~\cite{imhof:PRSB:2010}. 
In that regime, the population resolves to a pure (homogeneous) state between every two consecutive mutations.
When the space $\S$ of available strategies is finite, the long-term fate of the population can be characterized by a frequency vector $\freqs = (\freq_S\mid S\in\S)$ that records, for each strategy $S$, the relative proportion $\freq_S$ of time for which everyone in the population employs strategy $S$.
This frequency vector can be found in time that is polynomial in the population size $N$ and in the size $k$ of the strategy space,
by computing all the pairwise fixation probabilities $\{\pfix(S,S')\mid S,S'\in\S\}$
and solving for the stationary distribution of the underlying Markov chain~\cite{fudenberg2006imitation}.
Given the frequency vector, it is straightforward to compute the overall average cooperation rate $C$
as $C=\sum_{S\in\S} \freq_S\cdot \crate(S,S)$.

In the \SITwo, we show that for large population sizes $N$, each of the pairwise fixation probabilities $\pfix(S,S')$ is either very small (namely exponentially small in $N$), or quite large (namely at least inversely proportional to $N$), see~\thmfpin.
This allows us to argue that for the reduced strategy space $\S_3$, the cooperation rate $C$ tends to either 0 or to 1, as the population size $N$ grows large.
The intuition is that for large $N$, one of the entries $\freq_S$ in the frequency vector $\freqs$ tends to 1, whereas all other entries tend to $0$.
Thus, the overall cooperation rate tends to the pairwise cooperation rate $\crate(S,S)$ of this most frequent strategy against itself.
See~\thmunillstwo{} and \thmunillsthree{} for details.

\bigskip
\sne{Data and Code availability} 
All the datasets used in this paper together with the related computer code are available at \texttt{doi.org/10.6084/m9.figshare.21583554.v1}.

\bigskip
\sne{Acknoledgments}
C.H. acknowledges generous funding from the European Research Council (ERC) under the European Union's Horizon 2020 research and innovation program (Starting Grant 850529: E-DIRECT), and from the Max Planck Society.



\vspace{1cm}
\section*{\SIOne: Related literature}

In this Appendix we give a detailed account of the related literature.

\sne{Previous studies on the effect of mutations on direct reciprocity}
In this work, we show how mutation promotes cooperation in the infinitely repeated prisoner's dilemma. 
While there is by now a rich literature on the evolution of direct reciprocity and conditionally cooperative strategies~\cite{Glynatsi:HSSC:2021}, 
there has been little discussion on the quantitative effect of mutation rates. Instead, previous studies usually fall into one of the following four classes: 

\begin{enumerate}
\item {\it Studies that focus on evolutionary processes without any mutations}.
Some studies assume that individuals can only choose among a finite set of strategies (such as \allc, \alld, and \tft~\cite{brandt:JTB:2006,hilbe:BMB:2011}). Evolution among those finitely many strategies can be studied with the classical replicator equation~\cite{taylor:MB:1978}. Classical replicator dynamics describes a deterministic evolutionary process in infinitely large populations in the absence of mutations. The resulting dynamics critically depends on which strategies of direct reciprocity are considered. 
For example, between \alld{} and \gtft, replicator dynamics predicts a bistable competition~\cite{Rodriguez:JMB:2016}. For other strategy sets, researchers have observed cycles~\cite{nowak:AMC:1989} and the stable coexistence of several strategies~\cite{grujic:jtb:2012}. 

In addition to this work on well-mixed populations, researchers have also explored direct reciprocity in structured populations. One common assumption is that individuals are arranged on a regular lattice, and that they only interact with their immediate neighbors. Again, the respective simulations typically assume that evolution does not introduce any novel behaviors ~\cite{szabo:PRE:2000b,szolnoki:scirep:2014,szolnoki:pre:2014}. 

\item {\it Studies that assume mutations to be rare.} 
There is a considerable literature on evolution in the repeated prisoner's dilemma that assumes mutation rates to be vanishingly small. 
This assumption is made, for example, in studies of adaptive dynamics~\cite{geritz:EER:1998}. Originally, adaptive dynamics was introduced to explore evolution of continuous traits in infinite populations. This framework has been used, for example, to explore evolution of reactive or memory-1 strategies~\cite{nowak:AAM:1990,stewart:games:2015,Reiter:ncomms:2018}. In addition, Imhof and Nowak have extended adaptive dynamics to finite populations~\cite{imhof:PRSB:2010}, by exploiting the fact that the fixation probability of a single mutant in a homogeneous resident population can be computed explicitly~\cite{nowak:Nature:2004}. Since then, the rare-mutation assumption has been widely used to explore direct reciprocity in populations of finite size~\cite{hilbe:PNAS:2013,stewart:pnas:2013,stewart:pnas:2014,stewart:pnas:2016,Baek:SciRep:2016,Donahue:NComms:2020,Schmid:NHB:2021,park:NComms:2022,willensdorfer2005mutation,kurokawa:PRSB:2009,van-segbroeck:prl:2012,pinheiro:PLoSCB:2014}.

\item {\it Studies with a positive but constant mutation rate.} 
Especially some of the early and influential studies on direct reciprocity consider simulations with a strictly positive mutation rate~\cite{nowak:Nature:1992a,nowak:Nature:1993, hauert:PRSB:1997}. However, in many of those studies, the mutation rate has not been varied; the same mutation scheme is used throughout. 

\item {\it Studies that vary the mutation rate without finding a positive effect.}
In some cases, researchers have varied the mutation rate but without discovering a positive effect of intermediate mutation rates~\cite{Donahue:NComms:2020,Schmid:NHB:2021,park:NComms:2022,willensdorfer2005mutation}. 
Because the respective studies did not focus on the impact of mutations, they did not explore the parameter space systematically. 
In many cases, simulations have been run for parameter values that naturally favor cooperation even when mutations are rare.
\end{enumerate}

\noindent
In addition to the above studies, there is also important work that discusses the effect of mutations on (static) evolutionary stability. 
For example, in an early work of this kind, Boyd and Lorberbaum show that once one allows mutants to be heterogeneous, no pure strategy is evolutionarily stable~\cite{boyd:Nature:1987}. 
This result is based on the idea that for any resident population, it is possible to construct mixtures of mutant strategies that are favored in the respective environment. 
This line of research has been extended and further refined by subsequent studies~\cite{lorberbaum:JTB:1994,Lorberbaum:JTB:2002,van-veelen:PNAS:2012,garcia:jet:2016,Garcia:FRAI:2018}.
Importantly, however, this instability result is symmetric. 
It affects cooperative and non-cooperative equilibria alike. 
In contrast, we report that when evolutionary processes are simulated, mutations tend to have an asymmetric effect. 
They systematically favor the evolution of cooperative strategies.

~\\
\noindent
{\bf Previous studies on the effect of mutations on cooperation more generally.} 
In addition to this previous work on direct reciprocity there are some studies that stress the positive effect of mutations in other social dilemmas that are different from the infinitely repeated prisoner's dilemma. One example is the study by Traulsen et al~\cite{traulsen:PNAS:2009}, in which the authors explore the evolution of peer punishment. In their model, players first decide whether or not to contribute to a public good. In a second stage players then decide whether to punish non-contributors. Their baseline model considers three strategies: defectors, who neither contribute nor punish; cooperators who contribute to the public good but do not punish; and punishers who both contribute and punish. For this setup, the authors find that when mutations are rare, most players learn to defect. Here, punishers are dominated by cooperators, who are themselves dominated by defectors. However, once the mutation rate is sufficiently large, cooperators are shown to prevail. 

In another more recent study, Ram\'irez et al~\cite{Ramirez:Arxiv:2022} explore the evolutionary dynamics of the traveler's dilemma. In this game, two players need to choose an integer within an interval [$L,U$]. A player's payoff is given by the lower of the two integers. In addition, if the players' chosen integers are different, the player with the lower integer obtains an additional reward $R$. Traditional backward induction suggests that this game has a unique but inefficient equilibrium: both players choose $L$. However,  Ram\'irez et al find that players learn to choose large numbers when the reward is sufficiently small and when mutations are sufficiently abundant. 

McNamarra et al~\cite{mcnamara:Nature:2004} study evolution in a finitely repeated prisoner's dilemma. Their model is set up as follows. If both players mutually cooperate in all rounds, the total number of rounds is $N$. Otherwise, if any of the players defects before round $N$, the game stops after that defection. The rules of the game and the parameter $N$ is known to the players. Again, this game can be solved by backward induction. In the unique (but inefficient) equilibrium, both players defect immediately. For their evolutionary analysis, McNamarra et al consider strategies that are given by a threshold $n$. A player with threshold $n$ cooperates until round $n$ and defects in round $n\!+\!1$, unless the game stopped before.  Evolutionary simulations suggest that players choose a small threshold when mutations are rare, consistent with backward induction. Once mutations are common, the authors make the interesting observation that larger thresholds are favored (for the same reason that larger numbers become favored in the traveler's dilemma). Throughout their analysis, McNamarra et al use the same payoffs as Axelrod~\cite{axelrod:book:1984}. In particular, they do not explore the minimal $b/c$ ratio required for cooperation to emerge. 

The mechanism that leads to cooperation in the above studies is different from ours. 
In the above studies~\cite{traulsen:PNAS:2009,Ramirez:Arxiv:2022,mcnamara:Nature:2004}, full cooperation is generally unstable. 
The role of mutations there is to introduce strategies into the population that give a payoff advantage to cooperators. 
In contrast, in our setup, cooperation is generally stable. While there are equilibria in which everyone defects, there are also equilibria in which everyone cooperates. The role of mutations, in our study, is to help evolution to find these cooperative equilibria more efficiently, and to make them more robust against neutral invasions.\\

\noindent
{\bf Work that stresses the importance of variation and diversity.}
In addition to these previous studies on the effect of mutations, there is also work that stresses the positive impact of variation more generally. 
For example, models of (cultural) group selection show that cooperation can evolve once different groups compete~\cite{Boyd:HumEcol:1982}. 
When groups differ in how cooperative they are, the more cooperative groups may outcompete the less cooperative ones. 
Importantly, however, group selection is most effective when there is variation between groups but not within groups. 
Once individual groups tend to be heterogeneous, either because of migration or because of mutations, group selection typically fails to select for cooperation~\cite{traulsen:PNAS:2006}.
 
While the previously discussed studies consider mutations as an exogenous driver for variation, there is also work on how diversity can be maintained by natural selection (as opposed by mutation). For example, the interplay of diversity and cooperation was studied by Hauert and Doebeli~\cite{Hauert:PNAS:2021} in a different context: they explore how spatial adaptive dynamics leads to diversification in social dilemmas.\\

\noindent
{\bf Work that explores the effectiveness of direct reciprocity when individuals have more memory.}
In our work, we study a surprisingly simple mechanism for the evolution of direct reciprocity. One of our findings is that we get meaningful cooperation for $b/c$ close to one, which is an absolute threshold. 
Previous work suggests that direct reciprocity also becomes more effective when individuals recall more than the last round ~\cite{hauert:PRSB:1997,hauert:JTB:2002b,stewart:scirep:2016,hilbe:pnas:2017,Li:NatCS:2022}. 
An equilibrium analysis suggests that larger memory allows for strategies that can sustain full cooperation even as $b/c$ becomes small~\cite{hilbe:pnas:2017}. For $b/c=1.2$, the suggested strategy needs to memorize the outcome of the last five rounds.
However, as of now it is unclear whether evolution would effectively find such cooperative strategies.  
Already for 3-rounds memory, a systematic exploration of the strategy space is computationally infeasible due to the huge number of possible strategies~\cite{hilbe:pnas:2017}. 
As a result, it is not known to which extent higher memory can effectively promote the evolution of cooperation for small $b/c$ ratios.


\vfill
\pagebreak
\section*{\SITwo: Mathematical derivations}
\setcounter{section}{0}

In this Appendix we state and prove mathematical theorems that support the claims made in the main text.
First, we note that there is a formula for the cooperation rate and the payoff between two memory-1 strategies~\cite{press2012iterated}. For reference, we explicitly state the formula in the special case of reactive strategies.
Then, we use the formula to compute the overall cooperation rate in two natural cases, namely in the case $u=1$ and in the limit $u\to 0$.

\section*{Formula for the cooperation rate}
Recall that a stochastic memory-1 strategy is given by four parameters $p_{CC}$, $p_{CD}$, $p_{DC}$, $p_{DD}$ 
that denote the probabilities to cooperate after each of the possible outcomes of the last round,
$CC$, $CD$, $DC$, $DD$, respectively. 
Given two stochastic memory-1 strategies $S=(p_{CC},p_{CD},p_{DC},p_{DD})$ and $S'=(p'_{CC},p'_{CD},p'_{DC},p'_{DD})$, it is known how to compute both
the cooperation rate $\crate(S,S')$ for strategy $S$ when facing strategy $S'$,
and the corresponding payoff $\pi(S,S')$ derived from that interaction, see~\cite{press2012iterated}.
Recall that reactive strategies form a subset of memory-1 strategies.
A stochastic reactive strategy is given by two parameters $p$, $q$, where $p\in[0,1]$ (resp.\ $q\in[0,1]$) are the probabilities to cooperate in the next round, assuming that the co-player has just cooperated (resp.\ defected).
For the case of two stochastic reactive strategies $S=(p,q)$, the formula takes a simpler form, which we state here for reference.

\begin{lemma}\label{thm:crate} Let $S=(p,q)$ and $S'=(p',q')$.
If $S=S'\in\{\tft,\atft\}$ then $\crate(S,S')=1/2$, otherwise
\[\crate(S,S')= \frac{q'(p-q)+q}{1-(p-q)(p'-q')}.\]
Moreover,
\[\pi(S,S') = \crate(S',S) - c\cdot \crate(S,S').\]
\end{lemma}
\begin{proof} This is a special case of Eqn. [5] from~\cite{press2012iterated}.
\end{proof}

As a technical remark, we note that in order to completely define a stochastic reactive strategy for an infinitely repeated donation game, one should formally also define the probability $p_0$ to cooperate in the first round.
However, the formulas for cooperation rate and payoff do not depend on $p_0$ (except for two corner cases $p=1-q\in\{0,1\}$), hence we omit the parameter $p_0$ and identify the space of reactive strategies with points $(p,q)$ inside a unit square.
An alternative way to avoid these corner cases is to consider a tiny implementation error $\varepsilon$, and this is what we do in \SIFour{} when considering memory-2 strategies.

\section*{Case $u=1$}
When $u=1$ (a pure mutation process) the cooperation rate can be computed by averaging over possible pairs of strategies.
We note that the answer is independent of the population size $N$.

\begin{theorem}\label{thm:u1}
 Suppose $u=1$. Let $\S$ be any space of strategies (finite of infinite).
 Then
\[ \crate= \int_{S\in\S} \int_{S'\in \S} C(S,S') \d S \d S' . 
\]
\end{theorem}
\begin{proof} Consider a fixed pair of individuals. With probability $\d(S)$ the first one uses strategy $S$ and with probability $\d(S')$ the second one uses strategy $S'$. The claim follows by linearity of expectation.
\end{proof}

Note that when the strategy space $\S$ has finite size $k$, the integrals can be replaced by sums and we obtain that the cooperation rate $C$ is equal to the average entry of the corresponding $k\times k$ matrix of pairwise cooperation rates.
We also highlight another corollary of~\cref{thm:u1} 
that applies to the continuous strategy space $\S=[0,1]^2$. 

\begin{lemma}\label{thm:u1-cont} Let $u=1$ and $\S=[0,1]^2$ be the space of reactive strategies.
Then
$ \crate= \frac12$.
\end{lemma}
\begin{proof}
The proof is essentially by swapping the roles of ``cooperate'' and ``defect''.
Note that the function
\[f\colon (p,q,p',q')\mapsto (1-q,1-p,1-q',1-p')
\]
is a measure-preserving involution on $[0,1]^4$.
Thus it suffices to show that $\crate((p,q),(p',q')) + \crate((1-q,1-p),(1-q',1-p'))=1$.
Using~\cref{thm:crate}, this is straightforward algebra:
\begin{align*} \crate((p,q),(p',q')) + \crate((1-q,1-p),(1-q',1-p'))
&= \frac{q'(p-q)+q}{1-(p-q)(p'-q')} + \frac{(1-p')(p-q)+1-p}{1-(p-q)(p'-q')}\\
&= \frac{1+(1-p'+q'-1)(p-q)}{1-(p-q)(p'-q')}=1.\qedhere
\end{align*}
\end{proof}

\section*{Case $u\to 0$}
In the limit $u\to 0$ (rare mutation) we can assume that each mutation resolves before the new one occurs.
In this way we obtain the so-called Imhof-Nowak process~\cite{imhof:PRSB:2010}.
Given a cost $c\in(0,1)$, a population size $N$ and a strategy space $\S$,
we denote the corresponding cooperation rate by $\crate^\IN=\crate^\IN(c,N,\S)$.

Note that in general when the strategy space has finite size $k=|S|$,
the cooperation rate $\crate^\IN$ can be computed efficiently~\cite{fudenberg2006imitation}, that is,
in time that is polynomial in both $k$ and the population size $N$.
Below we prove two results which show that for certain 2- and 3-strategy systems (and large population sizes $N$) the cooperation rate $\crate^\IN$ tends either to 0 or to 1.
Intuitively, the reason is that when $N$ is large, then the Imhof-Nowak process spends majority of the time in a configuration where everyone in the population plays one fixed strategy from the available (finite) strategy space. The average overall cooperation rate is then determined by the cooperation rate of this single strategy against itself.
The key technical result is the following lemma.

\begin{lemma}\label{lem:fp-in}
Let
$\Pi=((x=\pi(R,R),y=\pi(R,M))
,(z=\pi(M,R),w=\pi(M,M)))$
be the $2\times 2$ payoff matrix for strategies $R$ (resident) and $M$ (mutant).
Consider the Imhof-Nowak process with population size $N$, a single mutant, and $N-1$ residents.
Recall that $\beta>0$ is the selection strength.
Then, as $N$ grows large, the fixation probability $\pfix(M,R)$ of the mutant satisfies the following:
\begin{enumerate}
\item If $x>z$ and $y>w$ then $\pfix(M,R)=e^{-\beta N\cdot c_1 +o(N)}$, where  $c_1=\frac12\big((x-z)+(y-w)\big)$.
\item If $x>z$ and $y\le w$ then $\pfix(M,R)=e^{-\beta N\cdot c_2 +o(N)}$, where  $c_2=\frac{(x-z)^2}{2(x-z+w-y)}$.
\item If $x<z$ and $x+w>z+y$ then $\pfix(M,R)=e^{-\beta N\cdot c_1 +o(N)}$, where $c_1=\frac12\big((x-z)+(y-w)\big)$.
\item If $x<z$ and $x+w\le z+y$ then $\pfix(M,R)=\Omega(1/N)$.
\end{enumerate}
\end{lemma}
\begin{proof} Let $\pim_k$ (resp. $\pir_k$) be the average payoff of a mutant (resp.\ resident) when there are precisely $k$ mutants. Let $p_k^{R\to M}$ (resp.\ $p_k^{M\to R}$) be the probability that, in a single step of the process, a single individual switches from the resident strategy to the mutant strategy (resp.\ vice versa). Then
\[ \gamma_k\stackrel{\operatorname{def}}{=}\frac{ p_k^{M\to R}}{ p_k^{R\to M}} =\frac{
\frac{k(N-k)}{N(N-1)} \cdot \frac{1}{1+e^{\beta(\pim_k-\pir_k)}}
}
{
\frac{(N-k)k}{N(N-1)} \cdot \frac{1}{1+e^{\beta(\pir_k-\pim_k)}}
}
=e^{ \beta\cdot (\pir_k-\pim_k) },
\]
and by a known formula for absorption time on 1-dimensional Markov chains we have
\begin{equation}\label{eqn:1-dim-formula}
\pfix(M,R) = \left(1 +  \sum_{i=1}^{N-1}  \left(\prod_{k=1}^i \gamma_k\right)\right)^{-1}
= \left(1+\gamma_1+\gamma_1\gamma_2+\dots + \gamma_1\gamma_2\dots\gamma_{N-1}\right)^{-1}.
\end{equation}
As $N$ grows large, the asymptotics of the right-hand side is determined by the asymptotics of the dominating term from among the $N$ terms $t_i=\prod_{k=1}^i \gamma_k$, for $i=0,1,\dots,N-1$.
Note that $\gamma_k>1$ if and only if $\pir_k>\pim_k$.
Let $\Delta_k=\pir_k-\pim_k$. By definition, we have
\begin{align*}
\pir_k &= \frac1{N-1}(k\cdot y + (N-k-1)\cdot x) = x + \frac k{N-1}(y-x)  \quad\text{and}\\
\pir_m &= \frac1{N-1}((k-1)\cdot w + (N-k)\cdot z) = z + \frac k{N-1}(w-z) +  \frac1{N-1}(z-w),
\end{align*}
and thus
\[\Delta_k=\pir_k-\pim_k = (x-z) + \frac k{N-1}(y-x+z-w) +\frac{w-z}{N-1}.
\]
In other words, for large $N$ we have that $\Delta_k=\pir_k-\pim_k$ is  a linear function of $k$ that satisfies
\[
\Delta_1 = x-z+\bigO(1/N)\quad\text{and}\quad \Delta_{N-1}=y-w+\bigO(1/N).
\]
Since $t_i=\prod_{k=1}^i \gamma_k = \exp(\beta\cdot( \sum_{k=1}^i \Delta_k))$, 
for all large enough $N$ there are 4 cases for which of the terms $t_i$ is dominating. Those cases depend primarily on whether $\Delta_1>0$.
\begin{enumerate}
\item Suppose that $x-z>0$ and $y-w>0$.
Then $\Delta_k>0$ for all $k=1,\dots,N-1$ (and all large enough $N$).
Thus $\gamma_i>1$ for all $i=1,\dots,N-1$ and the dominating term is the last one.
In that case we have
\[ \sum_{k=1}^{N-1} \Delta_k = (N-1)\cdot(x-z) + \frac{N}{2}\cdot\left(y- x+z-w\right) + (w-z) = N\cdot \frac{x-z+y-w}2 + o(N),
\]
and thus
$\pfix(M,R) =  e^{-\beta N\cdot c_1 + o(N)}$, 
 where $c_1=\frac12\big((x-z)+(y-w)\big)$.

\item Suppose that $x-z>0$ and $y-w\le 0$.
Then $\Delta_k>0$ if and only if $k<\frac{x-z}{(x-z)+(w-y)}N$.
The dominating term is then the one with $i^\star=\lfloor\frac{x-z}{(x-z)+(w-y)}N\rfloor$ and we have
\[ \sum_{k=1}^{i^\star} \Delta_k = i^\star\cdot (x-z) + \frac{(i^\star)^2}{2N} (y-x+z-w) + o(N)
 = N\cdot \frac{(x-z)^2}{2(x-z+w-y)} +o(N),
\]
and thus $\pfix(M,R) =  e^{-\beta N\cdot c_2 + o(N)}$, where $c_2=\frac{(x-z)^2}{2(x-z+w-y)}$.

\item Suppose that $x-z<0$ and $(x-z)+(y-w)>0$.
Then $\Delta_k>0$ if and only if $k>\frac{z-x}{(z-x)+(y-w)}N$. Since the last fraction is less than $1/2$ and the function $\Delta_k$ is linear in $k$,
the dominating term $t_i$ is the last one. As in the first case we conclude that
$\pfix(M,R) =  e^{-\beta N\cdot c_1 + o(N)}$, where $c_1=\frac12\big((x-z)+(y-w)\big)$.

\item Suppose that $x-z<0$ and $(x-z)+(y-w)\le 0$.
Then again $\Delta_k>0$ if and only if $k>\frac{z-x}{(z-x)+(y-w)}$.
However, this time the last fraction is at least $1/2$, so by linearity
all the terms (including the last term $t_{N-1}$) are $\bigO(1)$.
We thus have $\sum_{i=0}^{N-1} t_i\le N\cdot \bigO(1)=\bigO(N)$ and consequently $\pfix(M,R) =\Omega(1/N)$. \qedhere
\end{enumerate}
\end{proof}

\cref{lem:fp-in} allows us to efficiently compute the fixation probability $\pfix(M,R)$ of a mutant $M$ strategy invading a resident strategy $R$, when the population size $N$ is large.
In the limit of rare mutations (the Imhof-Nowak process), the population spends almost all the time in one of the homogeneous states.
Given a strategy $S\in\S$, let $\freq_S$ be the relative proportion of time that the population spends in the homogeneous state where all individuals play strategy $S$.
The frequency vector $\freqs=(\freq_S\mid S\in\S)$ can be computed from the pairwise fixation probabilities~\cite{fudenberg2006imitation}.
In particular, when there are only two strategies, say $A$ and $B$, then it is easy to see that
\[
\freq_A=\frac{\pfix(A,B)}{\pfix(A,B)+\pfix(B,A)}\quad\text{and}\quad \freq_B=\frac{\pfix(B,A)}{\pfix(A,B)+\pfix(B,A)}.
\]
For reader's convenience, we explicitly state the analogous result for strategy spaces consisting of three strategies.

\begin{lemma}\label{lem:mc-in-3}
Let $\S=\{A,B,C\}$ be a space of 3 strategies.
Let
\begin{align*}
 w(A) &= \pfix(A,B)\cdot \pfix(A,C) + \pfix(A,B)\cdot\pfix(B,C) + \pfix(A,C)\cdot\pfix(C,B),\\
 w(B) &= \pfix(B,C)\cdot \pfix(B,A) + \pfix(B,C)\cdot\pfix(C,A) + \pfix(B,A)\cdot\pfix(A,C),\\
 w(C) &= \pfix(C,A)\cdot \pfix(C,B) + \pfix(C,A)\cdot\pfix(A,B) + \pfix(C,B)\cdot\pfix(B,A),
\end{align*}
and $w=w(A)+w(B)+w(C)$. Then
$\freqs=\left( \freq_A,\freq_B,\freq_C\right) = \left( w(A)/w,  w(B)/w,  w(C)/w \right)$.
\end{lemma}

Using~\cref{lem:fp-in} one can show that, as the population size grows large, one of the frequencies in $\freqs=(\freq_S\mid S\in\S)$ typically tends to 1, whereas all the remaining ones tend to 0.
We show this explicitly for the 2-strategy systems $\S_q=\{\alld=(0,1),\gtft=(1,q)\}$ and for the 3-strategy system $\S_3=\{\alld=(0,1),\gtft=(1,q),\allc=(0,1)\}$.
See~\cref{fig:payoffs} for the respective pairwise payoffs.

\begin{theorem}\label{thm:u0-s2}
Fix $c\in(0,1)$, $q\in(0,1)$, and a 2-strategy system $\S_q=\{\alld=(0,1),\gtft=(1,q)\}$.
Let $q^\star=(1-c)/(1+c)$. Then, as $N\to\infty$, we have
$\crate^\IN(c,N,\S_q)\to 0$ if $q>q^\star$, and
 $\crate^\IN(c,N,\S_q)\to 1$ if $q<q^\star$.
\end{theorem}
\begin{proof}
We distinguish two cases.

First, suppose $q\ge 1-c$.
Then $\alld$ dominates $\gtft$, so clearly $\pfix(\alld,\gtft)\gg \pfix(\gtft,\alld)$ and 
$\crate^\IN(c,N,\S_q)\to 0$.

Second, suppose $q<1-c$.
Then $\pi(\alld,\alld)>\pi(\gtft,\alld)$ and $\pi(\alld,\gtft)<\pi(\gtft,\alld)$.
By~\cref{lem:fp-in}, case 2., we thus have $\pfix(\gtft,\alld)=e^{-\beta N\cdot c_G+o(N)}$,
where $c_G = \frac{(qc)^2}{2(1-q)(1-c)}$.
Likewise, by~\cref{lem:fp-in}, case 2., we have $\pfix(\alld,\gtft)=e^{-\beta N\cdot c_D+o(N)}$,
where $c_D=\frac{(1-c-q)^2}{2(1-q)(1-c)}$.
Note that the inequality $c_G>c_D$ is equivalent to $qc>1-c-q$ or $q>\frac{1-c}{1+c}=q^\star$.
For $q>q^\star$ we thus have $c_G>c_D$ and $\pfix(\gtft,\alld)\ll \pfix(\alld,\gtft)$,
implying that $\crate^\IN(c,N,\S_q)\to 0$.
In contrast, for $q<q^\star$ we have $c_G<c_D$, thus $\pfix(\gtft,\alld)\gg \pfix(\alld,\gtft)$, and
$\crate^\IN(c,N,\S_q)\to 1$.
\end{proof}

\begin{theorem}\label{thm:u0-s3}
Fix $c\in(0,1)$, $q\in(0,1)$, and a 3-strategy system $\S_3=\{\alld=(0,1),\gtft=(1,q),\allc=(0,1)\}$.
Then, as $N\to\infty$, we have $\crate^\IN(c,N,\S_3)\to 0$.
\end{theorem}
\begin{proof} 
 Note that from~\cref{thm:u0-s2} we have that both $\pfix(\alld,\gtft)$ and $\pfix(\gtft,\alld)$ are exponentially small (with possibly different constants in the exponents).
Moreover, since $\gtft$ and $\allc$ are neutral with respect to each other, we have
$\pfix(\gtft,\allc)=\pfix(\allc,\gtft)=1/N$.

Finally, consider $\alld$ and $\allc$.
By~\cref{lem:fp-in}, case 1., we have that $\pfix(\allc,\alld)$ is exponentially small.
In contrast, by~\cref{lem:fp-in}, case 4., we have that $\pfix(\alld,\allc)=\Omega(1/N)$.

Now we use~\cref{lem:mc-in-3}. Consider the nine terms used to express $w(\alld)$, $w(\gtft)$, $w(\allc)$.
The term $\pfix(\alld,\allc)\cdot\pfix(\allc,\gtft)$ is a product of two values, each of them of the order of at least $1/N$.
In contrast, in each of the remaining eight terms, at least one of the two values being multiplied is exponentially small.
Thus, in the limit $N\to\infty$, we have $\freq(\alld)\to 1$ and $\freq(\gtft),\freq(\allc)\to 0$,
implying that $\crate^\IN(c,N,\S_3)\to 0$.
\end{proof}


\newpage
\section*{\SIThree: The effect of phenotypic variation}
In the main text we consider mutations that cause heritable variation in behavior:
mutations introduce new strategies into the population and these new strategies can then be imitated by other population members.
We find that non-negligible levels of heritable variation help cooperation.
Among other reasons, they make populations of reciprocators more resilient against neutral invasion by unconditional cooperators
(only once there is variation, reciprocators and unconditional cooperators differ in their payoffs).
This mechanism, however, does not seem to require variation to be heritable.
Thus, it is natural to ask whether the observed effects persist when diversity is phenotypic (and thus non-heritable), as opposed to genotypic and heritable.
In this note we explore this question in detail.

Throughout this note, for simplicity we focus on the space of reactive strategies.
First, we describe how we extend our model to take into account phenotypic variation.
Then we present our analytical results, simulation results, and numerical results.
We find that the effects generally do persist with phenotypic (as opposed to genotypic) variation,
however their magnitude is generally smaller, and sometimes substantially smaller.

\subsection*{Modelling phenotypic variation}
To model phenotypic variation, on top of the previously defined mutation rate $u$ we consider an additional parameter $\mu\in[0,1]$, which we call the \textit{phenotypic mutation rate}.
The intuition is that with probability $\mu$ an individual plays a random strategy instead of their ``genetically prescribed'' strategy.
Formally, when interacting in a game, with probability $1-\mu$ the individual plays according to their actual strategy. With probability $\mu$, they instead use a random stochastic reactive strategy (with all stochastic reactive strategies being equally likely).
Thus, for $\mu=0$ we recover the standard case of no phenotypic variation, and for $\mu=1$ each individual always plays a random strategy (regardless of the heritable mutation rate $u$).
For intermediate $\mu\in(0,1)$, the payoffs derived from each interaction depend both on the ``genetically prescribed'' strategies of the two involved individuals, and on the phenotypic mutation rate $\mu$.

In order to determine the payoff $\pi_\mu(S,S')$ of an individual whose ``genetic'' information is a strategy $S$, when faced with an individual whose genetic information is a strategy $S'$, we need to consider four cases:
\begin{itemize}
\item With probability $(1-\mu)^2$ both individuals play the strategy given by their genotypes and the payoff is $\pi(S,S')$, which is given by~\cref{thm:crate}.
\item With probability $\mu\cdot (1-\mu)$ the first player plays a random reactive strategy and the co-player plays $S'$. Thus, we need to compute the payoff of a random reactive strategy against $S'$. We denote this $\pi(\cdot, S')$.
\item Similarly, with probability $(1-\mu)\cdot \mu$ the first player plays $S$ and the co-player plays a random reactive strategy, so we need to compute the payoff $\pi(S, \cdot)$ of $S$ against a random reactive strategy.
\item Finally, with probability $\mu^2$ both players play a reactive strategy selected uniformly at random, so we need to compute the payoff $\pi(\cdot,\cdot)$ of a random reactive strategy against a random reactive strategy.
\end{itemize}
In total, the overall payoff is then given by
\begin{equation}\label{eqn:pheno-combination}
\pi_\mu(S,S') = (1-\mu)^2\cdot \pi(S,S') + \mu(1-\mu)\cdot (\pi(\cdot,S')+\pi(S,\cdot)) + \mu^2\cdot \pi(\cdot,\cdot).
\end{equation}
We note that while the payoffs change, the underlying evolutionary dynamics based on pairwise comparison remains unchanged. In particular, we consider the phenotypic mutation rate $\mu$ in combination with the (genotypic) mutation rate $u$ that allows individuals to acquire a new, random (genotypic) strategy.

\subsection*{Results}
We present two types of results.

First, we find exact closed-form formulas for $\pi(S,\cdot)$, $\pi(\cdot,S')$, and $\pi(\cdot,\cdot)$.
By Equation~\ref{eqn:pheno-combination}, this immediately yields a closed-form formula for $\pi_\mu(S,S')$ and allows us to explore the setting of phenotypic variation by means of simulations and numerical calculations.

\begin{theorem} 
Let $S=(p,q)$ and $S'=(p',q')$ be reactive strategies and denote $r=q-p$ and $r'=q'-p'$.
Then
\begin{align*}
\pi(S,\cdot) &= (1-rc)\cdot G(p,q)-q/r,  \\
\pi(\cdot,S') &= (r'-c)\cdot G(p',q') - q'/r', \text{ and}\\
\pi(\cdot,\cdot) &= (1-c)/2,
\end{align*}
where $G(p,q)=\frac1{2r^3}\cdot\left( r^2 - (1-p-q)\cdot\left( (1-r)\log(1-r) + (1+r)\log(1+r) \right) \right)$.
\end{theorem}
\begin{proof} We have $\pi(S,\cdot)=\int_0^1\int_0^1 \pi(S,S') \d p' \d q'$, where
$\pi(S,S')$ is given by~\cref{thm:crate}.
This double integral can be computed using a standard mathematical software such as Mathematica. Similarly we compute $\pi(\cdot,S')=\int_0^1\int_0^1 \pi(S,S') \d p \d q$.
Finally, regarding $\pi(\cdot,\cdot)$ note that by~\cref{thm:u1-cont} the cooperation rate is $1/2$, and thus by~\cref{thm:crate} the payoff of the first player is $(1-c)/2$.
\end{proof}

Second, we present simulation results for the full space of reactive strategies, and numerical results for the reduced strategy set, see~\cref{fig:pheno}.
As stated above, we study the extent to which  phenotypic mutation rate $\mu$ can substitute for the (genotypic) mutation rate $u$.
In alignment with our expectations, we find that the key phenomena of the valley and the hump still occur. In particular, when the (genotypic) mutation rate $u$ is too small, cooperation can be promoted by increasing the phenotypic mutation rate $\mu$ instead of increasing the (genotypic) mutation rate $u$.
Interestingly, phenotypic mutation appears to be somewhat less effective and the magnitude of the effects is somewhat smaller, especially for high cost-to-benefit ratios.


\newpage 
\section*{\SIFour:  Evolutionary dynamics of memory-2 strategies}

So far, we have assumed that players use the most elementary strategies of reciprocity. 
The respective strategy spaces make minimal assumptions regarding the players' cognitive abilities. 
Players either only respond to the previous move of the co-player (reactive strategies), or they respond to the previous move of both players (memory-1 strategies). 
In the following, we explore whether mutations have a similarly positive effect when players have more than one-round memory. 
Because the dimensionality of the strategy space increases exponentially in the player's memory \cite{hilbe:pnas:2017},  we restrict ourselves to memory-2 strategies. 

\subsection*{Games among players with memory-2 strategies}

In infinitely repeated games, memory-2 strategies can be represented by a 16-dimensional vector,
\begin{equation} \label{Eq:Memory2Strategy}
\mathbf{p}\!=\!\!
\Big(
\minivector{p}{CC}{CC}, \minivector{p}{CC}{CD}, \minivector{p}{CD}{CC}, \minivector{p}{CD}{CD},
\minivector{p}{CC}{DC}, \minivector{p}{CC}{DD}, \minivector{p}{CD}{DC}, \minivector{p}{CD}{DD},
\minivector{p}{DC}{CC}, \minivector{p}{DC}{CD}, \minivector{p}{DD}{CC}, \minivector{p}{DD}{CD},
\minivector{p}{DC}{DC}, \minivector{p}{DC}{DD}, \minivector{p}{DD}{DC}, \minivector{p}{DD}{DD}
\Big).
\end{equation}
As usual, the entries reflect the player's conditional cooperation probabilities. The upper two indices of an entry represent the last two moves of the focal player (with the very last move coming first and the second-to last move coming second). The lower two indices represent the last two moves of the co-player. The space of memory-2 strategies trivially contains the set of reactive strategies and the set of memory-1 strategies. For example, within the space of memory-2 strategies, Tit-for-Tat takes the form
\begin{equation}
\mathbf{p} = (1,1,1,1,0,0,0,0,1,1,1,1,0,0,0,0). 
\end{equation}
Because each strategy is a 16-dimensional vector, the space of stochastic memory-2 strategies is difficult to explore exhaustively with simulations. In the following, we thus assume that individuals choose among the deterministic memory-2 strategies only. Therefore, all conditional cooperation probabilities are either zero or one. There are $2^{16}\!=\!65{,}536$ such strategies in total. To ensure that the long-term dynamics among pure memory-2 players is always well-defined, we assume players occasionally commit an implementation error. That is, with some probability $\varepsilon\!>\!0$, a player who wishes to cooperate instead defects by mistake (and vice versa). As a result, a player with strategy $\mathbf{p}$ effectively implements the strategy $(1\!-\!\varepsilon)\mathbf{p} + (1\!-\!\varepsilon)(\mathbf{1}-\mathbf{p})$. 

The payoffs of two memory-2 players can be computed with a Markov chain approach. To this end, suppose the effective strategies of player~1 and player~2 are $\mathbf{p}$ and~$\mathbf{q}$, respectively. 
The respective Markov chain has sixteen possible states, summarizing the last two moves of either player, $\minimatrix{CC}{CC},\minimatrix{CC}{CD},\minimatrix{CD}{CC},\ldots,\minimatrix{DD}{DD}$. Slightly abusing notation, here the upper indices refer to the past two actions of player~1 and the lower two indices refer to the past two actions of player~2. Given this ordering of the states, the transition matrix $M$ takes the following form,

\begin{equation*}
{
\scriptsize
\setlength{\arraycolsep}{0.2mm}
\left(
\begin{array}{cccccccccccccccc}
\minivector{p~\,}{CC}{CC}~\minivector{q~\,}{CC}{CC}	&0	&0	&0	
&\minivector{p~\,}{CC}{CC}~\minivector{\bar{q}~\,}{CC}{CC} &0	&0	&0
&\minivector{\bar{p}~\,}{CC}{CC}~\minivector{q~\,}{CC}{CC}	&0	&0	&0
&\minivector{\bar{p}~\,}{CC}{CC}~\minivector{\bar{q}~\,}{CC}{CC}	&0	&0	&0\\

\minivector{p~\,}{CC}{CD}~\minivector{q~\,}{CD}{CC}	&0	&0	&0	
&\minivector{p~\,}{CC}{CD}~\minivector{\bar{q}~\,}{CD}{CC} &0	&0	&0
&\minivector{\bar{p}~\,}{CC}{CD}~\minivector{q~\,}{CD}{CC}	&0	&0	&0
&\minivector{\bar{p}~\,}{CC}{CD}~\minivector{\bar{q}~\,}{CD}{CC}	&0	&0	&0\\

\minivector{p~\,}{CD}{CC}~\minivector{q~\,}{CC}{CD}	&0	&0	&0	
&\minivector{p~\,}{CD}{CC}~\minivector{\bar{q}~\,}{CC}{CD} &0	&0	&0
&\minivector{\bar{p}~\,}{CD}{CC}~\minivector{q~\,}{CC}{CD}	&0	&0	&0
&\minivector{\bar{p}~\,}{CD}{CC}~\minivector{\bar{q}~\,}{CC}{CD}	&0	&0	&0\\

\minivector{p~\,}{CD}{CD}~\minivector{q~\,}{CD}{CD}	&0	&0	&0	
&\minivector{p~\,}{CD}{CD}~\minivector{\bar{q}~\,}{CD}{CD} &0	&0	&0
&\minivector{\bar{p}~\,}{CD}{CD}~\minivector{q~\,}{CD}{CD}	&0	&0	&0
&\minivector{\bar{p}~\,}{CD}{CD}~\minivector{\bar{q}~\,}{CD}{CD}	&0	&0	&0\\

0	&\minivector{p~\,}{CC}{DC}~\minivector{q~\,}{DC}{CC}	&0	&0	
&0	&\minivector{p~\,}{CC}{DC}~\minivector{\bar{q}~\,}{DC}{CC}	&0	&0	
&0	&\minivector{\bar{p}~\,}{CC}{DC}~\minivector{q~\,}{DC}{CC}	&0	&0	
&0	&\minivector{\bar{p}~\,}{CC}{DC}~\minivector{\bar{q}~\,}{DC}{CC}	&0	&0\\

0	&\minivector{p~\,}{CC}{DD}~\minivector{q~\,}{DD}{CC}	&0	&0	
&0	&\minivector{p~\,}{CC}{DD}~\minivector{\bar{q}~\,}{DD}{CC}	&0	&0	
&0	&\minivector{\bar{p}~\,}{CC}{DD}~\minivector{q~\,}{DD}{CC}	&0	&0	
&0	&\minivector{\bar{p}~\,}{CC}{DD}~\minivector{\bar{q}~\,}{DD}{CC}	&0	&0\\

0	&\minivector{p~\,}{CD}{DC}~\minivector{q~\,}{DC}{CD}	&0	&0	
&0	&\minivector{p~\,}{CD}{DC}~\minivector{\bar{q}~\,}{DC}{CD}	&0	&0	
&0	&\minivector{\bar{p}~\,}{CD}{DC}~\minivector{q~\,}{DC}{CD}	&0	&0	
&0	&\minivector{\bar{p}~\,}{CD}{DC}~\minivector{\bar{q}~\,}{DC}{CD}	&0	&0\\

0	&\minivector{p~\,}{CD}{DD}~\minivector{q~\,}{DD}{CD}	&0	&0	
&0	&\minivector{p~\,}{CD}{DD}~\minivector{\bar{q}~\,}{DD}{CD}	&0	&0	
&0	&\minivector{\bar{p}~\,}{CD}{DD}~\minivector{q~\,}{DD}{CD}	&0	&0	
&0	&\minivector{\bar{p}~\,}{CD}{DD}~\minivector{\bar{q}~\,}{DD}{CD}	&0	&0\\

0	&0	&\minivector{p~\,}{DC}{CC}~\minivector{q~\,}{CC}{DC}	&0	
&0	&0	&\minivector{p~\,}{DC}{CC}~\minivector{\bar{q}~\,}{CC}{DC} &0	
&0	&0	&\minivector{\bar{p}~\,}{DC}{CC}~\minivector{q~\,}{CC}{DC}	&0	
&0	&0	&\minivector{\bar{p}~\,}{DC}{CC}~\minivector{\bar{q}~\,}{CC}{DC}	&0\\

0	&0	&\minivector{p~\,}{DC}{CD}~\minivector{q~\,}{CD}{DC}	&0	
&0	&0	&\minivector{p~\,}{DC}{CD}~\minivector{\bar{q}~\,}{CD}{DC} &0	
&0	&0	&\minivector{\bar{p}~\,}{DC}{CD}~\minivector{q~\,}{CD}{DC}	&0	
&0	&0	&\minivector{\bar{p}~\,}{DC}{CD}~\minivector{\bar{q}~\,}{CD}{DC}	&0\\

0	&0	&\minivector{p~\,}{DD}{CC}~\minivector{q~\,}{CC}{DD}	&0	
&0	&0	&\minivector{p~\,}{DD}{CC}~\minivector{\bar{q}~\,}{CC}{DD} &0	
&0	&0	&\minivector{\bar{p}~\,}{DD}{CC}~\minivector{q~\,}{CC}{DD}	&0	
&0	&0	&\minivector{\bar{p}~\,}{DD}{CC}~\minivector{\bar{q}~\,}{CC}{DD}	&0\\

0	&0	&\minivector{p~\,}{DD}{CD}~\minivector{q~\,}{CD}{DD}	&0	
&0	&0	&\minivector{p~\,}{DD}{CD}~\minivector{\bar{q}~\,}{CD}{DD} &0	
&0	&0	&\minivector{\bar{p}~\,}{DD}{CD}~\minivector{q~\,}{CD}{DD}	&0	
&0	&0	&\minivector{\bar{p}~\,}{DD}{CD}~\minivector{\bar{q}~\,}{CD}{DD}	&0\\

0	&0	&0	&\minivector{p~\,}{DC}{DC}~\minivector{q~\,}{DC}{DC}	
&0	&0	&0	&\minivector{p~\,}{DC}{DC}~\minivector{\bar{q}~\,}{DC}{DC} 	
&0	&0	&0	&\minivector{\bar{p}~\,}{DC}{DC}~\minivector{q~\,}{DC}{DC}	
&0	&0	&0	&\minivector{\bar{p}~\,}{DC}{DC}~\minivector{\bar{q}~\,}{DC}{DC}\\

0	&0	&0	&\minivector{p~\,}{DC}{DD}~\minivector{q~\,}{DD}{DC}	
&0	&0	&0	&\minivector{p~\,}{DC}{DD}~\minivector{\bar{q}~\,}{DD}{DC} 	
&0	&0	&0	&\minivector{\bar{p}~\,}{DC}{DD}~\minivector{q~\,}{DD}{DC}	
&0	&0	&0	&\minivector{\bar{p}~\,}{DC}{DD}~\minivector{\bar{q}~\,}{DD}{DC}\\

0	&0	&0	&\minivector{p~\,}{DD}{DC}~\minivector{q~\,}{DC}{DD}	
&0	&0	&0	&\minivector{p~\,}{DD}{DC}~\minivector{\bar{q}~\,}{DC}{DD} 	
&0	&0	&0	&\minivector{\bar{p}~\,}{DD}{DC}~\minivector{q~\,}{DC}{DD}	
&0	&0	&0	&\minivector{\bar{p}~\,}{DD}{DC}~\minivector{\bar{q}~\,}{DC}{DD}\\

0	&0	&0	&\minivector{p~\,}{DD}{DD}~\minivector{q~\,}{DD}{DD}	
&0	&0	&0	&\minivector{p~\,}{DD}{DD}~\minivector{\bar{q}~\,}{DD}{DD} 	
&0	&0	&0	&\minivector{\bar{p}~\,}{DD}{DD}~\minivector{q~\,}{DD}{DD}	
&0	&0	&0	&\minivector{\bar{p}~\,}{DD}{DD}~\minivector{\bar{q}~\,}{DD}{DD}\\
\end{array}
\right).
}
\end{equation*}
Here, we have used the notation $\minivector{\bar{p}}{ij}{kl} = 1\!-\!\minivector{p}{ij}{kl}$, and similarly $\minivector{\bar{q}}{ij}{kl} = 1\!-\!\minivector{q}{ij}{kl}$ for $i,j,k,l\!\in\!\{C,D\}$. 
Given the above transition matrix, we compute the (normalized) invariant distribution of the respective Markov chain by solving $v \!=\! v M$. This invariant distribution is a 16-dimensional vector,
\begin{equation} \label{Eq:Memory2Invariant}
\mathbf{v}\!=\!\!
\Big(
\minivector{v}{CC}{CC}, \minivector{v}{CC}{CD}, \minivector{v}{CD}{CC}, \minivector{v}{CD}{CD},
\minivector{v}{CC}{DC}, \minivector{v}{CC}{DD}, \minivector{v}{CD}{DC}, \minivector{v}{CD}{DD},
\minivector{v}{DC}{CC}, \minivector{v}{DC}{CD}, \minivector{v}{DD}{CC}, \minivector{v}{DD}{CD},
\minivector{v}{DC}{DC}, \minivector{v}{DC}{DD}, \minivector{v}{DD}{DC}, \minivector{v}{DD}{DD}
\Big).
\end{equation}

\noindent
Using this invariant distribution, we calculate how often each player cooperates on average. To this end, we sum up over all outcomes in which the player cooperates in the last round,
\begin{equation} \label{Eq:CooperationMem2}
\begin{array}{l}
C(\mathbf{p},\mathbf{q})= \minivector{v}{CC}{CC} \!+\! \minivector{v}{CC}{CD} \!+\! \minivector{v}{CD}{CC} \!+\! \minivector{v}{CD}{CD} \!+\!
\minivector{v}{CC}{DC} \!+\! \minivector{v}{CC}{DD} \!+\! \minivector{v}{CD}{DC} \!+\!  \minivector{v}{CD}{DD},\\
C(\mathbf{q},\mathbf{p})=  \minivector{v}{CC}{CC} \!+\! \minivector{v}{CC}{CD} \!+\! \minivector{v}{CD}{CC} \!+\! \minivector{v}{CD}{CD} \!+\!
\minivector{v}{DC}{CC}\!+\! \minivector{v}{DC}{CD}\!+\! \minivector{v}{DD}{CC}\!+\! \minivector{v}{DD}{CD}.
\end{array}
\end{equation}
Given these average cooperation rates, the players' payoffs are
\begin{equation} \label{Eq:PayoffM2}
\pi(\mathbf{p},\mathbf{q}) = b\!\cdot\!C(\mathbf{q},\mathbf{p}) - c\!\cdot\!C(\mathbf{p},\mathbf{q})~~~~~\text{and}~~~~~
\pi(\mathbf{q},\mathbf{p}) = b\!\cdot\!C(\mathbf{p},\mathbf{q}) - c\!\cdot\!C(\mathbf{q},\mathbf{p}).
\end{equation}

\subsection*{Evolutionary dynamics among players with memory-2 strategies}

To explore the dynamics among players with memory-2 strategies, we again use individual-based simulations. 
The simulations follow the same protocol as before. 
There is a finite population of size $N$. 
Each player is equipped with a memory-2 strategy. 
In each generation, players engage with all other population members in a repeated prisoner's dilemma. 
For each pairwise game, we compute the players' payoffs according to the previous subsection, using~\eqref{Eq:PayoffM2}. 
As a result of all these pairwise interactions, players receive an average payoff per pairwise interaction. 
This average payoff depends on the player's strategy and on the strategy of all other population members. 

After these pairwise interactions, players update their strategies according to a pairwise imitation process. 
Each player is given an opportunity to update its strategy. 
For this update, there are two possibilities. 
With probability $u$ there is a mutation. In that case, the player adopts a randomly chosen deterministic memory-2 strategy. 
Otherwise the focal player chooses a random role model from the population. 
The probability that the focal player switches to the role model's strategy is again determined by a Fermi function, as described in the main text. 
After every player had a chance to update its strategy, the above process is repeated for many generations. 
We use computer simulations to explore this evolutionary process. 

In Figure~\ref{fig:memtwo} we show the results of these simulations. 
There, we simultaneously vary how costly cooperation is (on the $x$-axis), and how often mutations occur (on the $y$-axis). 
The colors of the contour plot indicate how often individuals cooperate on average. 
The results exhibit similar patterns as in the case of reactive strategies and memory-1 strategies. 
Again, if the costs $c$ are very small, most individuals learn to cooperate for both small and intermediate mutation rates. 
If the costs $c$ are exceedingly large, individuals defect for both small and intermediate mutation rates. 
In between,  we observe mutations to be beneficial. 
Individuals are cooperative for intermediate mutation rates, but they tend to defect when mutations are rare. 

These simulations suggest that the results we have presented in the main text are not restricted to individuals with one-round memory. 
Instead we observe similar characteristic curves as before. 
Especially when cooperation costs are substantial, it again takes sizeable mutation rates to establish cooperation.


\renewcommand{\thefigure}{S\arabic{figure}}
\setcounter{figure}{0}

\clearpage
\newpage
\section*{Supplementary figures}
\vspace{3cm}

\begin{figure}[h]
\centering
\includegraphics[scale=1.25]{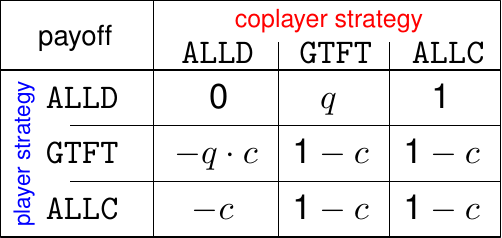}
\caption{Pairwise payoffs among the strategies $\alld=(0,0)$, $\gtft=(1,q)$ and $\allc=(1,1)$, in a donation game with normalized benefit $b=1$ and cost $c\in(0,1)$.}
\label{fig:payoffs}
\end{figure}

\clearpage
\newpage

\begin{figure}[t!] 
  \centering
   \includegraphics[scale=1]{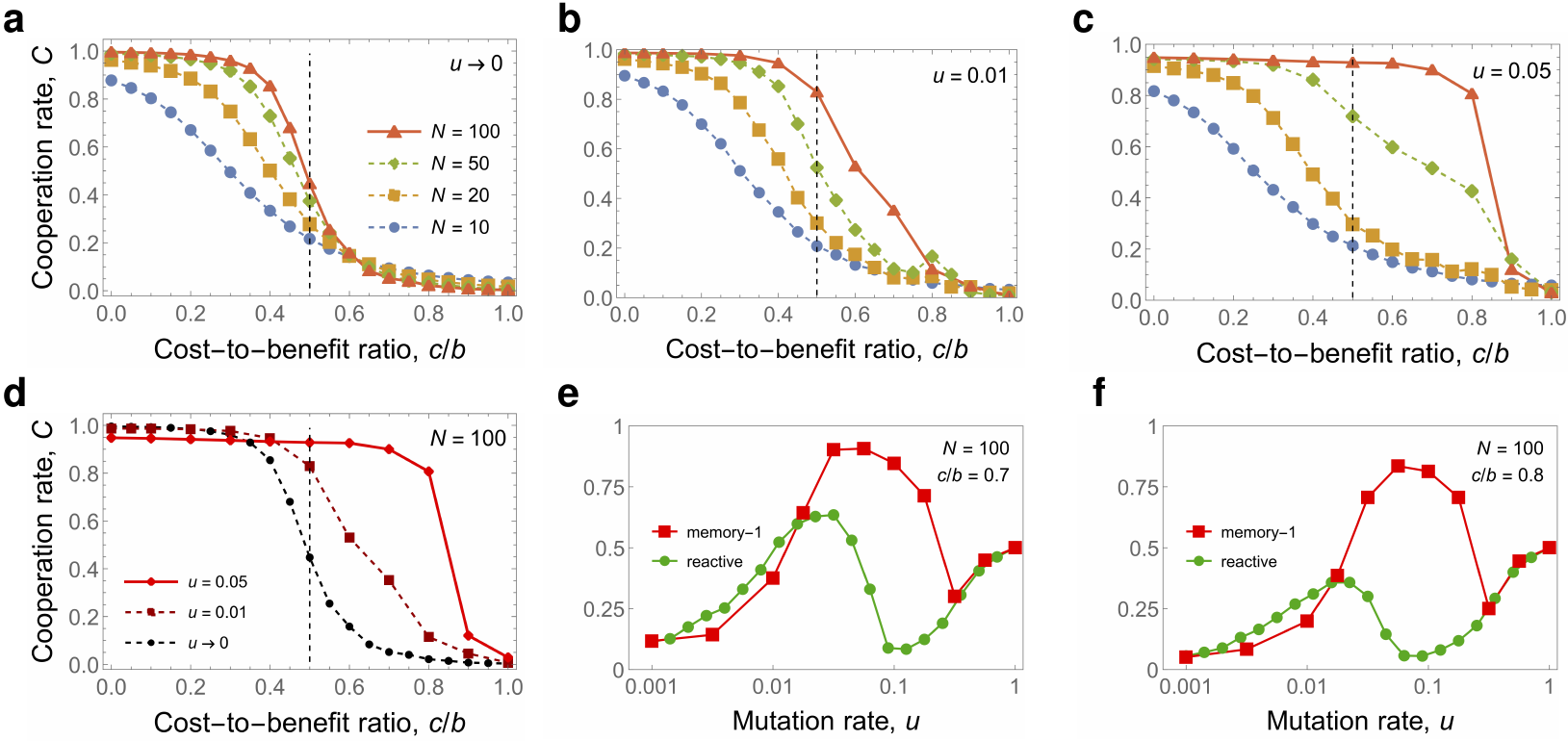}
   
\caption{
\textbf{The space of memory-1 strategies.}
The observed effects persist when we move from the space of reactive strategies to the space of memory-1 strategies. Memory-1 strategies take into account the actions of both players in the previous round.  Each memory-1 strategy is given by four probabilities. \textbf{a,} In the limit $u\to 0$, we observe no substantial cooperation $c/b>1/2$.
\textbf{b--d,} In contrast, for $u=0.01$ and $u=0.05$ we observe substantial cooperation even for $c/b>0.5$, especially when $N=100$ (compare to~\figMeanCooperationPartTwo{} from the main text).
As a function of $u$, the cooperation rate~$C$ again exhibits both the valley and the hump, both for \textbf{e,} $c/b=0.7$ and \textbf{f,} $c/b=0.8$.
}
\label{fig:f0-c-plot-m1}
\end{figure}

\clearpage
\newpage

\begin{figure}[t] 
  \centering
   \includegraphics[width=0.6\linewidth]{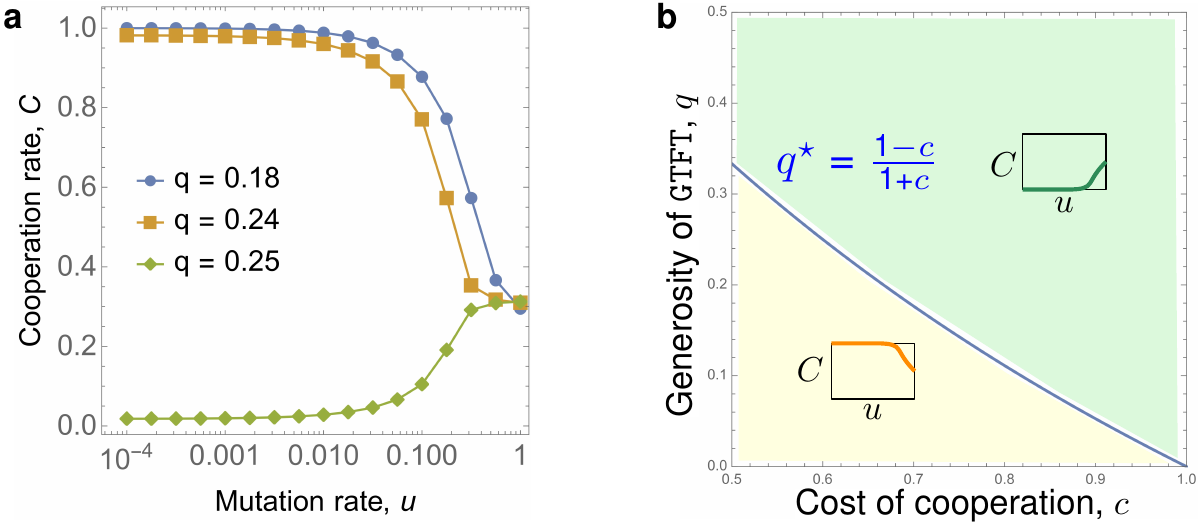}
\caption{
\textbf{The two-strategy system $\{\alld,\gtft\}$ has no valley and no hump.}
 \textbf{a,}
 If the mutating individual adopts either $\alld$ or $\gtft=(1,q)$ with equal probability, the cooperation rate $C$ is monotone with respect to the mutation rate $u$.
 Parameters: $c=0.6$ and $N=100$.
\textbf{b,} The threshold value of $q$ that determines whether $C$
 tends to 0 or 1 as $u\to 0$
 is $q^\star=(1-c)/(1+c)$, see~\cref{thm:u0-s2} in the \SITwo.
}
\label{fig:2-hump}
\end{figure}

\clearpage
\newpage

\begin{figure}[t] 
  \centering
\includegraphics[width=0.4\linewidth]{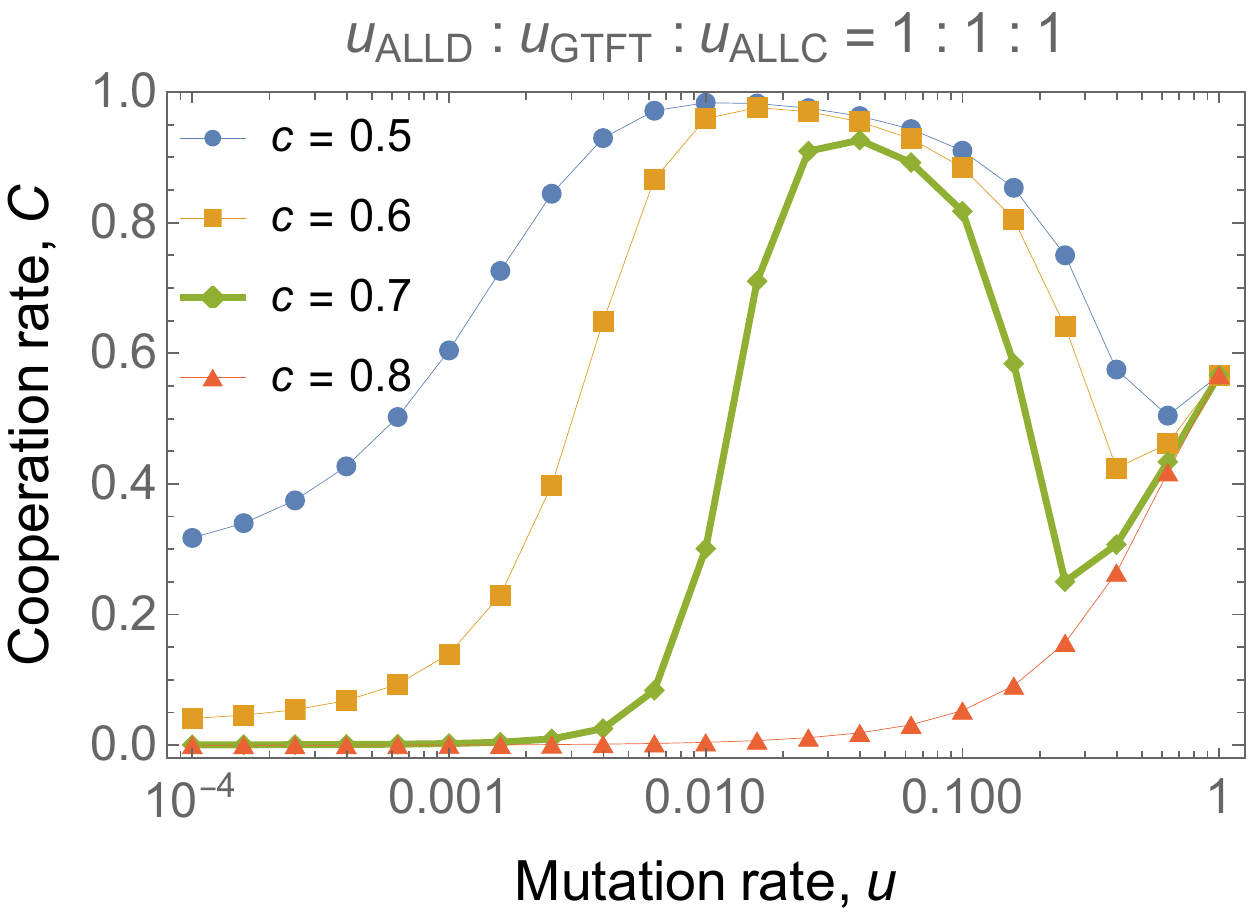}
\includegraphics[width=0.4\linewidth]{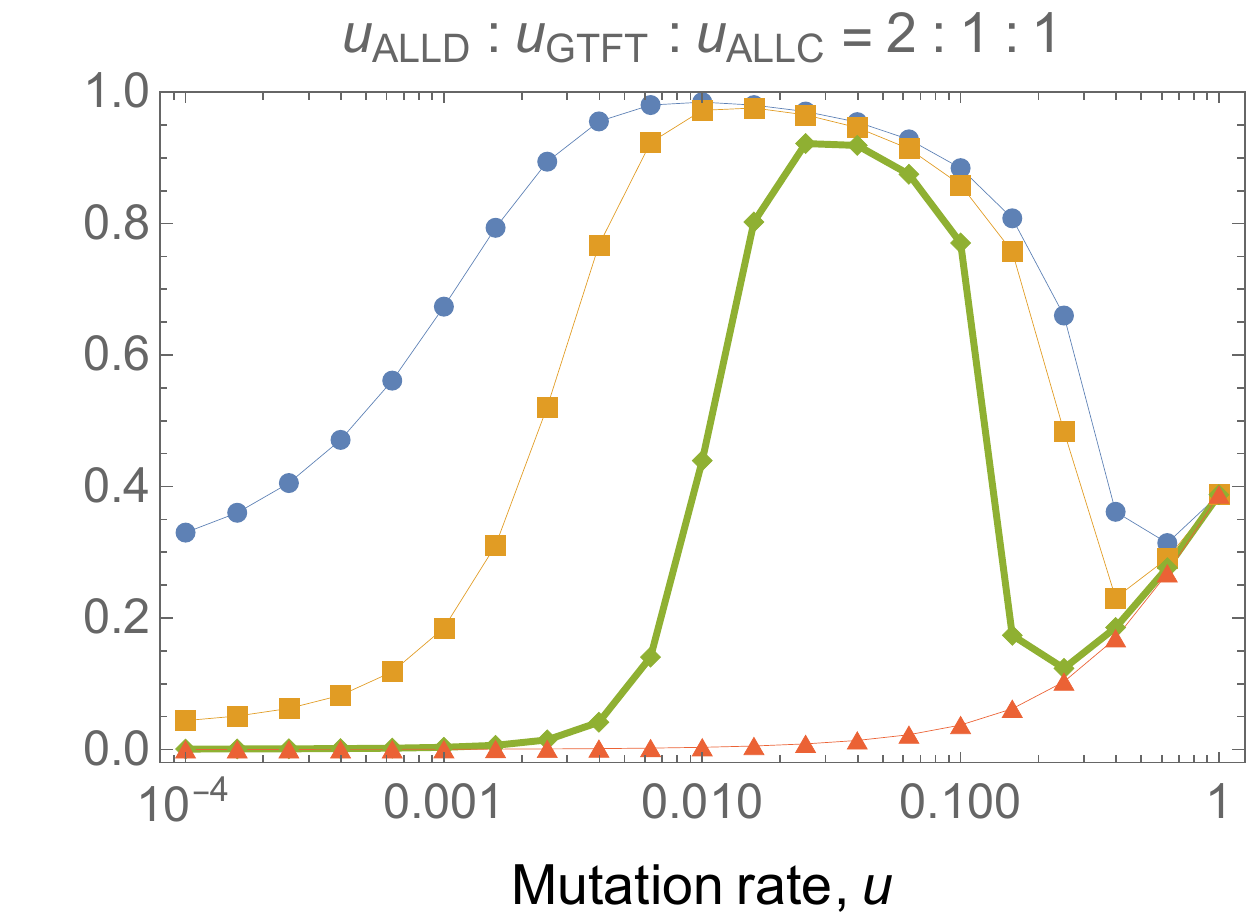}
\caption{
\textbf{The three-strategy system $\{\alld,\gtft,\allc\}$ has both valley and hump.}
\textbf{a,} Mutation chooses each of the three strategies with the same probability.
\textbf{b,} Mutation chooses $\alld$ with probability $0.5$, and each of $\gtft$, $\allc$ with probability $0.25$.
Parameters $q=0.1$ and $N=100$.
}
\label{fig:3-hump}
\end{figure}

\clearpage
\newpage

\begin{figure}[t] 
  \centering
   \includegraphics[width=0.9\linewidth]{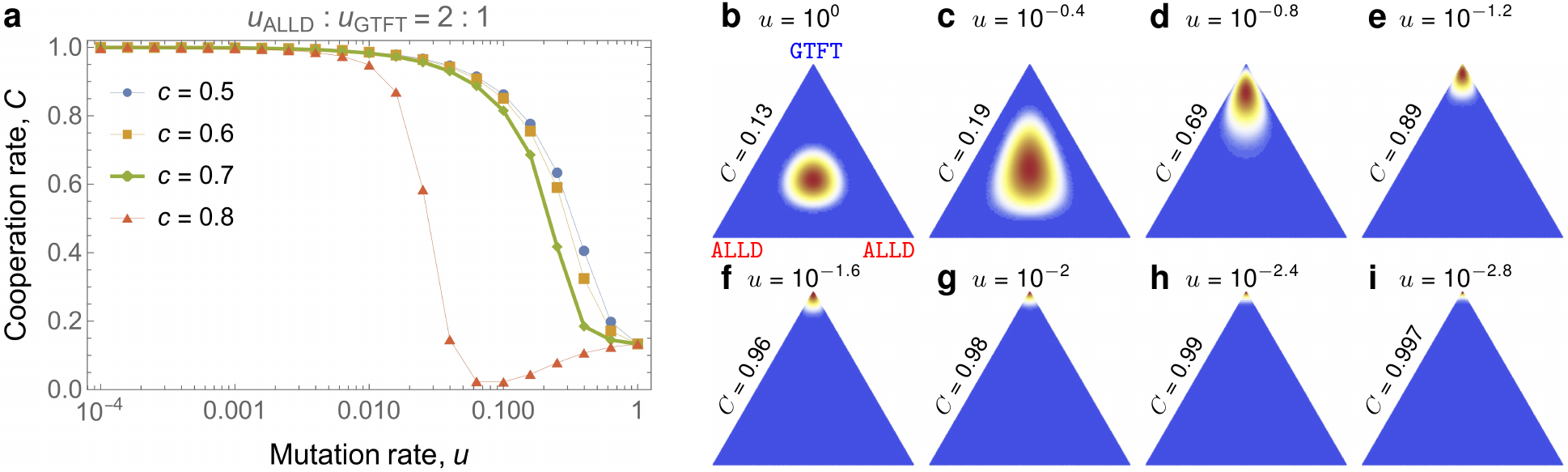}
\caption{
\textbf{Heat maps for the three strategy system $\{\alld,\gtft,\alld\}$.}
Replacing $\allc$ with another copy of $\alld$ in the 3-strategy system  is equivalent to considering a 2-strategy system, where mutants select $\alld$ with probability $2/3$.
\textbf{a,} The cooperation rate tends to 1 as $u\to 0$. Thus, in the original system $\{\alld,\gtft,\allc\}$ the $\allc$ players catalyse the transition from the cooperative equilibrium to the defective one, as $u\to 0$.
\textbf{b-i,} Relative frequencies of the population composition (red is high).
Here again $\gtft=[1,0.1]$, $N=100$, $c=0.7$, and $u$ varies from \textbf{b,} $u=10^0$ to \textbf{i,} $u=10^{-2.8}$.
}
\label{fig:3-temp-ddg}
\end{figure}

\clearpage
\newpage

\begin{figure}[t] 
  \centering
   \includegraphics[width=0.75\linewidth]{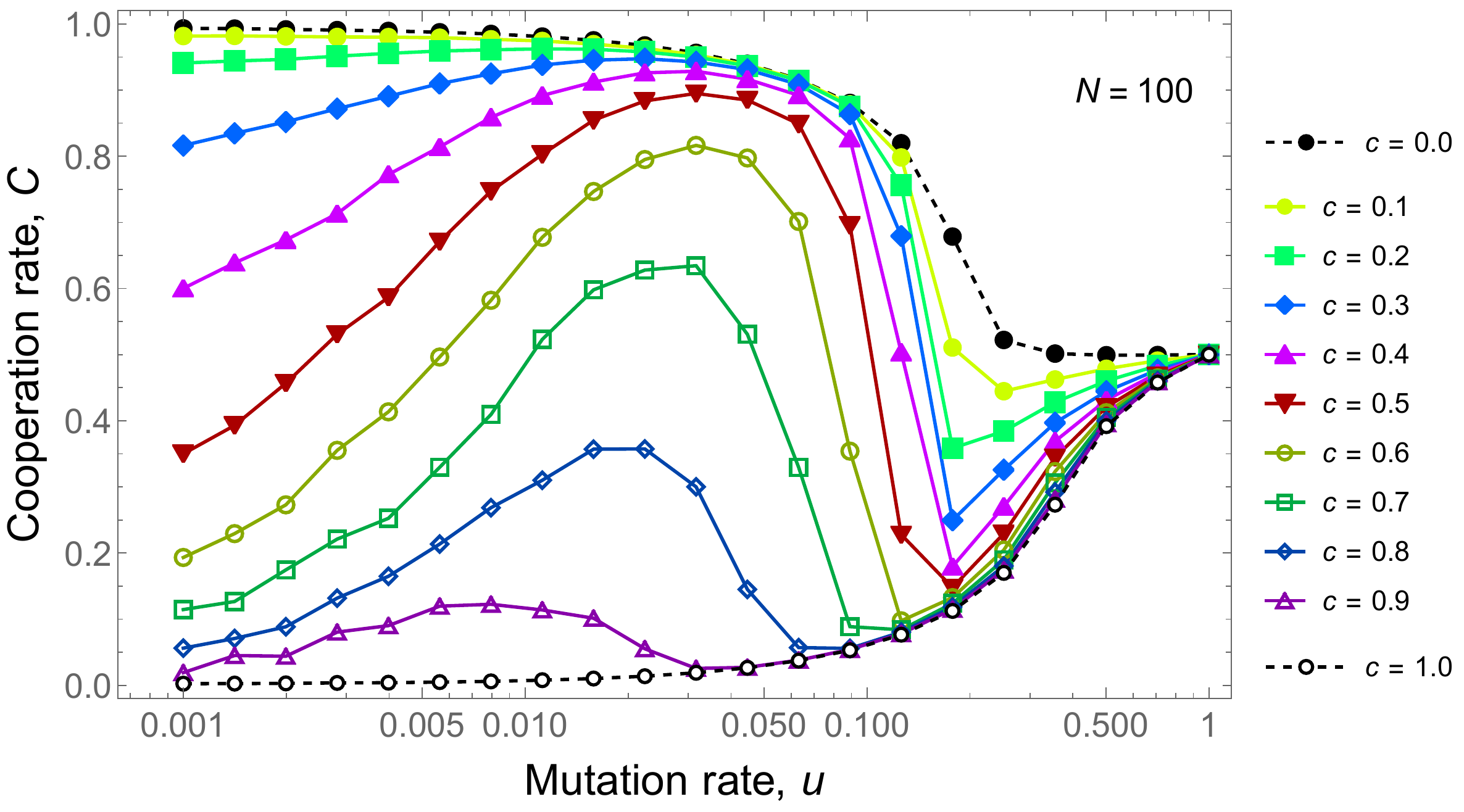}
\caption{
\textbf{Optimum diversity for reactive strategies and all costs $c$.}
The average cooperation rate, $C$, is shown as a function of the mutation rate, $u$, for  various $c$ values.  Simulations are run for at least $10^9$ updates to get reliable averages. Parameters: $N=100$, $b=1$, $\beta=10$.
}
\label{fig:uall-call}
\end{figure}

\clearpage
\newpage

\begin{figure}[t] 
  \centering
   \includegraphics[width=0.75\linewidth]{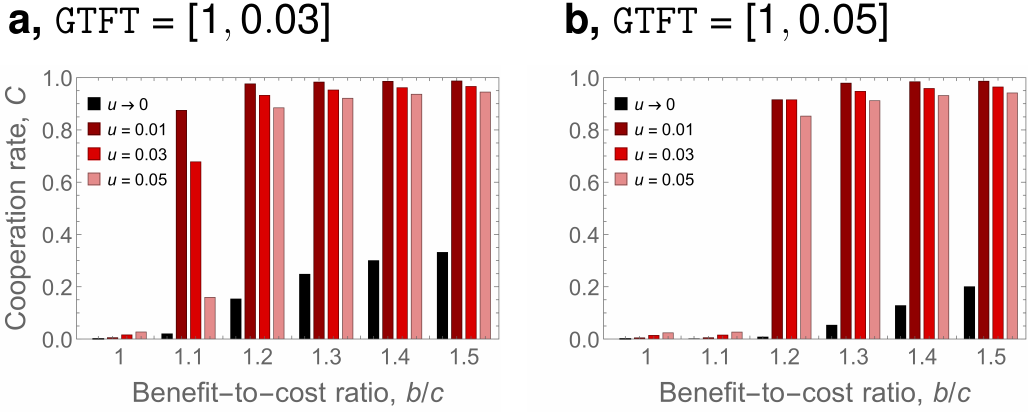}
\caption{
\textbf{The diversity effect on cooperation in the reduced strategy set $\{\alld,\gtft,\allc\}$.}
As for the set of all reactive strategies or all memory-1 strategies, adding mutations (red bars) substantially enhances cooperation compared to the limit of rare mutation $u\to 0$ (black bars).
Parameters: $N=100$, $\beta=10$.
}
\label{fig:si-bars}
\end{figure}

\clearpage
\newpage

\begin{figure}[t]
\centering
\includegraphics[scale=1]{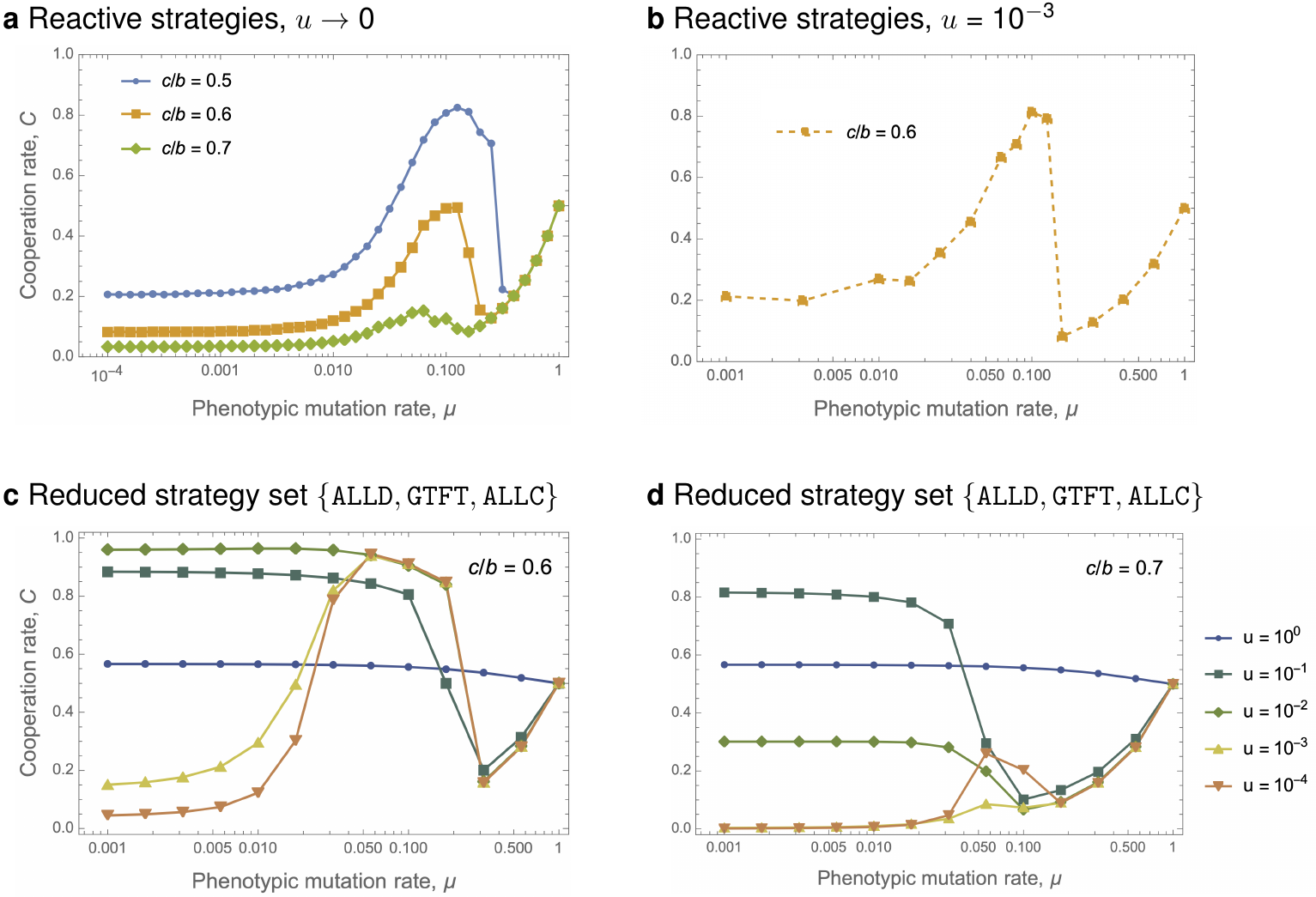}
\caption{\textbf{The effect of phenotypic variation.} The phenotypic mutation rate $\mu$ is the probability that an individual plays a random stochastic reactive strategy rather than their ``genetically prescribed'' strategy.
\textbf{a,} For the set of all reactive strategies, in the limit of rare (genotypic) mutation rate ($u\to 0$) the cooperation rate exhibits both the valley and the hump, as a function of the phenotypic mutation rate $\mu$. However, the hump is not as tall, especially when $c/b=0.7$. Here $N=100$, simulations are run for $10^9$ attempted invasions, and we show median of 3 runs.
\textbf{b,} Similar results hold for fixed positive genotypic mutation rate $u=10^{-3}$. Here again $N=100$, simulations are run for $10^9$ generations, and we show median of 3 runs.
\textbf{c-d,} For the reduced strategy set $\mathcal{S}=\{\alld,\gtft,\allc\}$ and small cost-to-benefit ratio $c/b=0.6$, phenotypic variation can successfully substitute for the genetic variation, achieving high cooperation rates. With higher cost-to-benefit ratio $c/b=0.7$, this is no longer true as adding phenotypic variation never increases the cooperation rate above 0.5. Here $N=100$, and the plotted values are obtained by numerically solving the underlying Markov chain.
}
\label{fig:pheno}
\end{figure}

\clearpage
\newpage

\begin{figure}[h!]
\centering
\includegraphics[width=0.5\linewidth]{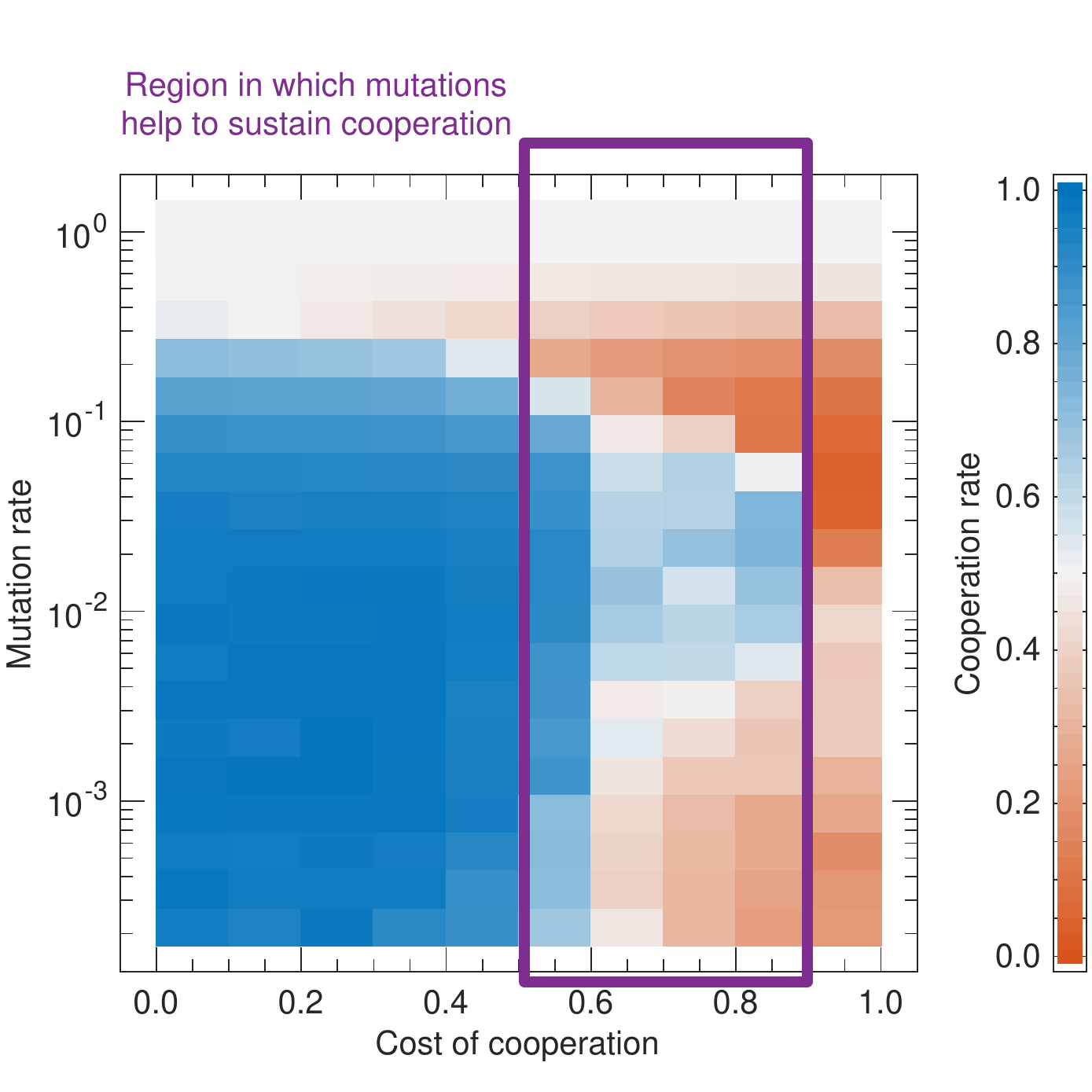}
\caption{\textbf{Evolutionary dynamics of memory-2 strategies.} 
Using the same evolutionary process as before, we explore the dynamics when individuals can choose among all deterministic memory-2 strategies. 
For our simulations, we vary the cooperation costs $c\!\in\!\{0.05, 0.15, \ldots, 0.95\}$, depicted on the $x$-axis, and the mutation rate $u\!\in\!\{10^{-3.6}, 10^{-3.4}, \ldots, 10^0\}$, depicted on the $y$-axis. 
For each parameter combination, we run the evolutionary process for $10^8$ generations. 
The plot shows the average cooperation rate in the second half of each simulation run (to account for transient effects in the early stage of the simulation). 
Qualitatively, we observe three different regions. 
For small cooperation costs (in the left part of the panel), cooperation evolves even when mutations are rare. 
For exceedingly high cooperation costs (in the right part of the panel), cooperation hardly evolves at all. 
However, in between, we observe high cooperation rates but only for intermediate mutation rates. 
This region is highlighted by a red frame. 
Within this frame, we again observe the characteristic valley and the hump. 
Parameters: Similar to before, simulations are run for a population of size $N\!=\!100$ and a selection strength $\beta\!=\!10$. To ensure that payoffs are well-defined, strategies are misimplemented with a small error rate $\varepsilon\!=\!0.001$. Each point depicts an average of 35 independent simulations for the respective parameter combination.}
\label{fig:memtwo}
\end{figure}

\bibliographystyle{naturemag}
\bibliography{bibliography}

\begin{thebibliography}{10}
\expandafter\ifx\csname url\endcsname\relax
  \def\url#1{\texttt{#1}}\fi
\expandafter\ifx\csname urlprefix\endcsname\relax\def\urlprefix{URL }\fi
\providecommand{\bibinfo}[2]{#2}
\providecommand{\eprint}[2][]{\url{#2}}

\bibitem{nowak:Science:2006}
\bibinfo{author}{Nowak, M.~A.}
\newblock \bibinfo{title}{Five rules for the evolution of cooperation}.
\newblock \emph{\bibinfo{journal}{Science}} \textbf{\bibinfo{volume}{314}},
  \bibinfo{pages}{1560--1563} (\bibinfo{year}{2006}).

\bibitem{skyrms2014evolution}
\bibinfo{author}{Skyrms, B.}
\newblock \emph{\bibinfo{title}{Evolution of the social contract}}
  (\bibinfo{publisher}{Cambridge University Press}, \bibinfo{year}{2014}).

\bibitem{trivers:QRB:1971}
\bibinfo{author}{Trivers, R.~L.}
\newblock \bibinfo{title}{The evolution of reciprocal altruism}.
\newblock \emph{\bibinfo{journal}{The Quarterly Review of Biology}}
  \textbf{\bibinfo{volume}{46}}, \bibinfo{pages}{35--57}
  (\bibinfo{year}{1971}).

\bibitem{axelrod:book:1984}
\bibinfo{author}{Axelrod, R.}
\newblock \emph{\bibinfo{title}{The evolution of cooperation}}
  (\bibinfo{publisher}{Basic Books}, \bibinfo{address}{New York, NY},
  \bibinfo{year}{1984}).

\bibitem{boyd:Nature:1987}
\bibinfo{author}{Boyd, R.} \& \bibinfo{author}{Lorberbaum, J.}
\newblock \bibinfo{title}{No pure strategy is evolutionary stable in the
  iterated prisoner's dilemma game}.
\newblock \emph{\bibinfo{journal}{Nature}} \textbf{\bibinfo{volume}{327}},
  \bibinfo{pages}{58--59} (\bibinfo{year}{1987}).

\bibitem{boyd:JTB:1989}
\bibinfo{author}{Boyd, R.}
\newblock \bibinfo{title}{Mistakes allow evolutionary stability in the repeated
  {P}risoner's {D}ilemma game}.
\newblock \emph{\bibinfo{journal}{Journal of Theoretical Biology}}
  \textbf{\bibinfo{volume}{136}}, \bibinfo{pages}{47--56}
  (\bibinfo{year}{1989}).

\bibitem{Kraines:TaD:1989}
\bibinfo{author}{Kraines, D.~P.} \& \bibinfo{author}{Kraines, V.~Y.}
\newblock \bibinfo{title}{Pavlov and the prisoner's dilemma}.
\newblock \emph{\bibinfo{journal}{Theory and Decision}}
  \textbf{\bibinfo{volume}{26}}, \bibinfo{pages}{47--79}
  (\bibinfo{year}{1989}).

\bibitem{nowak:Nature:1993}
\bibinfo{author}{Nowak, M.~A.} \& \bibinfo{author}{Sigmund, K.}
\newblock \bibinfo{title}{A strategy of win-stay, lose-shift that outperforms
  tit-for-tat in the {P}risoner's {D}ilemma game}.
\newblock \emph{\bibinfo{journal}{Nature}} \textbf{\bibinfo{volume}{364}},
  \bibinfo{pages}{56--58} (\bibinfo{year}{1993}).

\bibitem{killingback:PRSB:1999}
\bibinfo{author}{Killingback, T.}, \bibinfo{author}{Doebeli, M.} \&
  \bibinfo{author}{Knowlton, N.}
\newblock \bibinfo{title}{Variable investment, the continuous prisoner's
  dilemma, and the origin of cooperation}.
\newblock \emph{\bibinfo{journal}{Proceedings of the Royal Society B}}
  \textbf{\bibinfo{volume}{266}}, \bibinfo{pages}{1723--1728}
  (\bibinfo{year}{1999}).

\bibitem{killingback:AmNat:2002}
\bibinfo{author}{Killingback, T.} \& \bibinfo{author}{Doebeli, M.}
\newblock \bibinfo{title}{The continuous {P}risoner's {D}ilemma and the
  evolution of cooperation through reciprocal altruism with variable
  investment}.
\newblock \emph{\bibinfo{journal}{The American Naturalist}}
  \textbf{\bibinfo{volume}{160}}, \bibinfo{pages}{421--438}
  (\bibinfo{year}{2002}).

\bibitem{fischer2013fusing}
\bibinfo{author}{Fischer, I.} \emph{et~al.}
\newblock \bibinfo{title}{Fusing enacted and expected mimicry generates a
  winning strategy that promotes the evolution of cooperation}.
\newblock \emph{\bibinfo{journal}{Proceedings of the National Academy of
  Sciences}} \textbf{\bibinfo{volume}{110}}, \bibinfo{pages}{10229--10233}
  (\bibinfo{year}{2013}).

\bibitem{garcia:jet:2016}
\bibinfo{author}{Garc\'ia, J.} \& \bibinfo{author}{van Veelen, M.}
\newblock \bibinfo{title}{In and out of equilibrium {I}: {E}volution of
  strategies in repeated games with discounting}.
\newblock \emph{\bibinfo{journal}{Journal of Economic Theory}}
  \textbf{\bibinfo{volume}{161}}, \bibinfo{pages}{161--189}
  (\bibinfo{year}{2016}).

\bibitem{akin:JDG:2017}
\bibinfo{author}{Akin, E.}
\newblock \bibinfo{title}{Good strategies for the iterated prisoner's dilemma:
  {S}male vs. {M}arkov}.
\newblock \emph{\bibinfo{journal}{Journal of Dynamics and Games}}
  \textbf{\bibinfo{volume}{4}}, \bibinfo{pages}{217--253}
  (\bibinfo{year}{2017}).

\bibitem{Hilbe:NHB:2018}
\bibinfo{author}{Hilbe, C.}, \bibinfo{author}{Chatterjee, K.} \&
  \bibinfo{author}{Nowak, M.~A.}
\newblock \bibinfo{title}{Partners and rivals in direct reciprocity}.
\newblock \emph{\bibinfo{journal}{Nature Human Behaviour}}
  \textbf{\bibinfo{volume}{2}}, \bibinfo{pages}{469--477}
  (\bibinfo{year}{2018}).

\bibitem{Glynatsi:SciRep:2020}
\bibinfo{author}{Glynatsi, N.} \& \bibinfo{author}{Knight, V.}
\newblock \bibinfo{title}{Using a theory of mind to find best responses to
  memory-one strategies}.
\newblock \emph{\bibinfo{journal}{Scientific Reports}}
  \textbf{\bibinfo{volume}{10}}, \bibinfo{pages}{1--9} (\bibinfo{year}{2020}).

\bibitem{Glynatsi:HSSC:2021}
\bibinfo{author}{Glynatsi, N.} \& \bibinfo{author}{Knight, V.}
\newblock \bibinfo{title}{A bibliometric study of research topics,
  collaboration and centrality in the field of the {I}terated {P}risoner's
  {D}ilemma}.
\newblock \emph{\bibinfo{journal}{Humanities and {S}ocial {S}ciences
  {C}ommunications}} \textbf{\bibinfo{volume}{8}}, \bibinfo{pages}{45}
  (\bibinfo{year}{2021}).

\bibitem{sigmund:book:2010}
\bibinfo{author}{Sigmund, K.}
\newblock \emph{\bibinfo{title}{The Calculus of Selfishness}}
  (\bibinfo{publisher}{Princeton Univ. Press}, \bibinfo{address}{Princeton,
  NJ}, \bibinfo{year}{2010}).

\bibitem{rapoport:book:1965}
\bibinfo{author}{Rapoport, A.} \& \bibinfo{author}{Chammah, A.~M.}
\newblock \emph{\bibinfo{title}{Prisoner's Dilemma}}
  (\bibinfo{publisher}{University of Michigan Press, Ann Arbor},
  \bibinfo{year}{1965}).

\bibitem{nowak:book:2006}
\bibinfo{author}{Nowak, M.~A.}
\newblock \emph{\bibinfo{title}{Evolutionary dynamics}}
  (\bibinfo{publisher}{Harvard University Press}, \bibinfo{address}{Cambridge
  MA}, \bibinfo{year}{2006}).

\bibitem{molander:jcr:1985}
\bibinfo{author}{Molander, P.}
\newblock \bibinfo{title}{The optimal level of generosity in a selfish,
  uncertain environment}.
\newblock \emph{\bibinfo{journal}{Journal of Conflict Resolution}}
  \textbf{\bibinfo{volume}{29}}, \bibinfo{pages}{611--618}
  (\bibinfo{year}{1985}).

\bibitem{nowak:Nature:1992a}
\bibinfo{author}{Nowak, M.~A.} \& \bibinfo{author}{Sigmund, K.}
\newblock \bibinfo{title}{Tit for tat in heterogeneous populations}.
\newblock \emph{\bibinfo{journal}{Nature}} \textbf{\bibinfo{volume}{355}},
  \bibinfo{pages}{250--253} (\bibinfo{year}{1992}).

\bibitem{friedman:RES:1971}
\bibinfo{author}{Friedman, J.}
\newblock \bibinfo{title}{A non-cooperative equilibrium for supergames}.
\newblock \emph{\bibinfo{journal}{Review of Economic Studies}}
  \textbf{\bibinfo{volume}{38}}, \bibinfo{pages}{1--12} (\bibinfo{year}{1971}).

\bibitem{stewart:pnas:2014}
\bibinfo{author}{Stewart, A.~J.} \& \bibinfo{author}{Plotkin, J.~B.}
\newblock \bibinfo{title}{Collapse of cooperation in evolving games}.
\newblock \emph{\bibinfo{journal}{Proceedings of the National Academy of
  Sciences USA}} \textbf{\bibinfo{volume}{111}}, \bibinfo{pages}{17558 --
  17563} (\bibinfo{year}{2014}).

\bibitem{nowak:Nature:2004}
\bibinfo{author}{Nowak, M.~A.}, \bibinfo{author}{Sasaki, A.},
  \bibinfo{author}{Taylor, C.} \& \bibinfo{author}{Fudenberg, D.}
\newblock \bibinfo{title}{Emergence of cooperation and evolutionary stability
  in finite populations}.
\newblock \emph{\bibinfo{journal}{Nature}} \textbf{\bibinfo{volume}{428}},
  \bibinfo{pages}{646--650} (\bibinfo{year}{2004}).

\bibitem{Garcia:FRAI:2018}
\bibinfo{author}{Garc\'ia, J.} \& \bibinfo{author}{van Veelen, M.}
\newblock \bibinfo{title}{No strategy can win in the repeated prisoner's
  dilemma: {L}inking game theory and computer simulations}.
\newblock \emph{\bibinfo{journal}{Frontiers in Robotics and AI}}
  \textbf{\bibinfo{volume}{5}}, \bibinfo{pages}{102} (\bibinfo{year}{2018}).

\bibitem{Hindersin:SciRep:2019}
\bibinfo{author}{Hindersin, L.}, \bibinfo{author}{Wu, B.},
  \bibinfo{author}{Traulsen, A.} \& \bibinfo{author}{Garc\'ia, J.}
\newblock \bibinfo{title}{Computation and simulation of evolutionary game
  dynamics in finite populations}.
\newblock \emph{\bibinfo{journal}{Scientific Reports}}
  \textbf{\bibinfo{volume}{9}}, \bibinfo{pages}{6946} (\bibinfo{year}{2019}).

\bibitem{fudenberg2006imitation}
\bibinfo{author}{Fudenberg, D.} \& \bibinfo{author}{Imhof, L.~A.}
\newblock \bibinfo{title}{Imitation processes with small mutations}.
\newblock \emph{\bibinfo{journal}{Journal of Economic Theory}}
  \textbf{\bibinfo{volume}{131}}, \bibinfo{pages}{251--262}
  (\bibinfo{year}{2006}).

\bibitem{wu:JMB:2012}
\bibinfo{author}{Wu, B.}, \bibinfo{author}{Gokhale, C.~S.},
  \bibinfo{author}{Wang, L.} \& \bibinfo{author}{Traulsen, A.}
\newblock \bibinfo{title}{How small are small mutation rates?}
\newblock \emph{\bibinfo{journal}{Journal of Mathematical Biology}}
  \textbf{\bibinfo{volume}{64}}, \bibinfo{pages}{803--827}
  (\bibinfo{year}{2012}).

\bibitem{mcavoy:jet:2015}
\bibinfo{author}{McAvoy, A.}
\newblock \bibinfo{title}{Comment on ``{I}mitation processes with small
  mutations''}.
\newblock \emph{\bibinfo{journal}{J. Econ. Theory}}
  \textbf{\bibinfo{volume}{159}}, \bibinfo{pages}{66--69}
  (\bibinfo{year}{2015}).

\bibitem{nowak:AAM:1990}
\bibinfo{author}{Nowak, M.~A.} \& \bibinfo{author}{Sigmund, K.}
\newblock \bibinfo{title}{The evolution of stochastic strategies in the
  prisoner's dilemma}.
\newblock \emph{\bibinfo{journal}{Acta Applicandae Mathematicae}}
  \textbf{\bibinfo{volume}{20}}, \bibinfo{pages}{247--265}
  (\bibinfo{year}{1990}).

\bibitem{stewart:games:2015}
\bibinfo{author}{Stewart, A.~J.} \& \bibinfo{author}{Plotkin, J.~B.}
\newblock \bibinfo{title}{The evolvability of cooperation under local and
  non-local mutations}.
\newblock \emph{\bibinfo{journal}{Games}} \textbf{\bibinfo{volume}{6}},
  \bibinfo{pages}{231--250} (\bibinfo{year}{2015}).

\bibitem{Reiter:ncomms:2018}
\bibinfo{author}{Reiter, J.~G.}, \bibinfo{author}{Hilbe, C.},
  \bibinfo{author}{Rand, D.~G.}, \bibinfo{author}{Chatterjee, K.} \&
  \bibinfo{author}{Nowak, M.~A.}
\newblock \bibinfo{title}{Crosstalk in concurrent repeated games impedes direct
  reciprocity and requires stronger levels of forgiveness}.
\newblock \emph{\bibinfo{journal}{Nature Communications}}
  \textbf{\bibinfo{volume}{9}}, \bibinfo{pages}{555} (\bibinfo{year}{2018}).

\bibitem{hilbe:PNAS:2013}
\bibinfo{author}{Hilbe, C.}, \bibinfo{author}{Nowak, M.~A.} \&
  \bibinfo{author}{Sigmund, K.}
\newblock \bibinfo{title}{The evolution of extortion in iterated prisoner's
  dilemma games}.
\newblock \emph{\bibinfo{journal}{Proceedings of the National Academy of
  Sciences USA}} \textbf{\bibinfo{volume}{110}}, \bibinfo{pages}{6913--6918}
  (\bibinfo{year}{2013}).

\bibitem{stewart:pnas:2013}
\bibinfo{author}{Stewart, A.~J.} \& \bibinfo{author}{Plotkin, J.~B.}
\newblock \bibinfo{title}{From extortion to generosity, evolution in the
  iterated prisoner's dilemma}.
\newblock \emph{\bibinfo{journal}{Proceedings of the National Academy of
  Sciences USA}} \textbf{\bibinfo{volume}{110}}, \bibinfo{pages}{15348--15353}
  (\bibinfo{year}{2013}).

\bibitem{stewart:pnas:2016}
\bibinfo{author}{Stewart, A.~J.}, \bibinfo{author}{Parsons, T.~L.} \&
  \bibinfo{author}{Plotkin, J.~B.}
\newblock \bibinfo{title}{Evolutionary consequences of behavioral diversity}.
\newblock \emph{\bibinfo{journal}{Proceedings of the National Academy of
  Sciences USA}} \textbf{\bibinfo{volume}{113}}, \bibinfo{pages}{E7003--E7009}
  (\bibinfo{year}{2016}).

\bibitem{Donahue:NComms:2020}
\bibinfo{author}{Donahue, K.}, \bibinfo{author}{Hauser, O.},
  \bibinfo{author}{Nowak, M.} \& \bibinfo{author}{Hilbe, C.}
\newblock \bibinfo{title}{Evolving cooperation in multichannel games}.
\newblock \emph{\bibinfo{journal}{Nature Communications}}
  \textbf{\bibinfo{volume}{11}}, \bibinfo{pages}{3885} (\bibinfo{year}{2020}).

\bibitem{Schmid:NHB:2021}
\bibinfo{author}{Schmid, L.}, \bibinfo{author}{Chatterjee, K.},
  \bibinfo{author}{Hilbe, C.} \& \bibinfo{author}{Nowak, M.}
\newblock \bibinfo{title}{A unified framework of direct and indirect
  reciprocity}.
\newblock \emph{\bibinfo{journal}{Nature Human Behaviour}}
  \textbf{\bibinfo{volume}{5}}, \bibinfo{pages}{1292--1302}
  (\bibinfo{year}{2021}).

\bibitem{park:NComms:2022}
\bibinfo{author}{Park, P.~S.}, \bibinfo{author}{Nowak, M.~A.} \&
  \bibinfo{author}{Hilbe, C.}
\newblock \bibinfo{title}{Cooperation in alternating interactions with memory
  constraints -- source code and data}.
\newblock \emph{\bibinfo{journal}{Nature Communications}}
  \textbf{\bibinfo{volume}{13}}, \bibinfo{pages}{737} (\bibinfo{year}{2022}).

\bibitem{kurokawa:PRSB:2009}
\bibinfo{author}{Kurokawa, S.} \& \bibinfo{author}{Ihara, Y.}
\newblock \bibinfo{title}{Emergence of cooperation in public goods games}.
\newblock \emph{\bibinfo{journal}{Proceedings of the Royal Society B}}
  \textbf{\bibinfo{volume}{276}}, \bibinfo{pages}{1379--1384}
  (\bibinfo{year}{2009}).

\bibitem{van-segbroeck:prl:2012}
\bibinfo{author}{van Segbroeck, S.}, \bibinfo{author}{Pacheco, J.~M.},
  \bibinfo{author}{Lenaerts, T.} \& \bibinfo{author}{Santos, F.~C.}
\newblock \bibinfo{title}{Emergence of fairness in repeated group
  interactions}.
\newblock \emph{\bibinfo{journal}{Physical Review Letters}}
  \textbf{\bibinfo{volume}{108}}, \bibinfo{pages}{158104}
  (\bibinfo{year}{2012}).

\bibitem{pinheiro:PLoSCB:2014}
\bibinfo{author}{Pinheiro, F.~L.}, \bibinfo{author}{Vasconcelos, V.~V.},
  \bibinfo{author}{Santos, F.~C.} \& \bibinfo{author}{Pacheco, J.~M.}
\newblock \bibinfo{title}{Evolution of all-or-none strategies in repeated
  public goods dilemmas}.
\newblock \emph{\bibinfo{journal}{PLoS Comput Biol}}
  \textbf{\bibinfo{volume}{10}}, \bibinfo{pages}{e1003945}
  (\bibinfo{year}{2014}).

\bibitem{traulsen:PNAS:2010}
\bibinfo{author}{Traulsen, A.}, \bibinfo{author}{Semmann, D.},
  \bibinfo{author}{Sommerfeld, R.~D.}, \bibinfo{author}{Krambeck, H.-J.} \&
  \bibinfo{author}{Milinski, M.}
\newblock \bibinfo{title}{Human strategy updating in evolutionary games}.
\newblock \emph{\bibinfo{journal}{Proceedings of the National Academy of
  Sciences USA}} \textbf{\bibinfo{volume}{107}}, \bibinfo{pages}{2962--2966}
  (\bibinfo{year}{2010}).

\bibitem{grujic:SciRep:2014}
\bibinfo{author}{Grujic, J.} \emph{et~al.}
\newblock \bibinfo{title}{A comparative analysis of spatial prisoner's dilemma
  experiments: Conditional cooperation and payoff irrelevance}.
\newblock \emph{\bibinfo{journal}{Scientific Reports}}
  \textbf{\bibinfo{volume}{4}}, \bibinfo{pages}{4615} (\bibinfo{year}{2014}).

\bibitem{Baek:SciRep:2016}
\bibinfo{author}{Baek, S.~K.}, \bibinfo{author}{Jeong, H.~C.},
  \bibinfo{author}{Hilbe, C.} \& \bibinfo{author}{Nowak, M.~A.}
\newblock \bibinfo{title}{Comparing reactive and memory-one strategies of
  direct reciprocity}.
\newblock \emph{\bibinfo{journal}{Scientific Reports}}
  \textbf{\bibinfo{volume}{6}}, \bibinfo{pages}{25676} (\bibinfo{year}{2016}).

\bibitem{press2012iterated}
\bibinfo{author}{Press, W.~H.} \& \bibinfo{author}{Dyson, F.~J.}
\newblock \bibinfo{title}{Iterated prisoner's dilemma contains strategies that
  dominate any evolutionary opponent}.
\newblock \emph{\bibinfo{journal}{Proceedings of the National Academy of
  Sciences}} \textbf{\bibinfo{volume}{109}}, \bibinfo{pages}{10409--10413}
  (\bibinfo{year}{2012}).

\bibitem{blume:GEB:1993}
\bibinfo{author}{Blume, L.~E.}
\newblock \bibinfo{title}{The statistical mechanics of strategic interaction}.
\newblock \emph{\bibinfo{journal}{Games and Economic Behavior}}
  \textbf{\bibinfo{volume}{5}}, \bibinfo{pages}{387--424}
  (\bibinfo{year}{1993}).

\bibitem{szabo:PRE:1998}
\bibinfo{author}{Szab{\'o}, G.} \& \bibinfo{author}{T{\H o}ke, C.}
\newblock \bibinfo{title}{Evolutionary {P}risoner's {D}ilemma game on a square
  lattice}.
\newblock \emph{\bibinfo{journal}{Physical Review E}}
  \textbf{\bibinfo{volume}{58}}, \bibinfo{pages}{69--73}
  (\bibinfo{year}{1998}).

\bibitem{traulsen:PRE:2006b}
\bibinfo{author}{Traulsen, A.}, \bibinfo{author}{Nowak, M.~A.} \&
  \bibinfo{author}{Pacheco, J.~M.}
\newblock \bibinfo{title}{Stochastic dynamics of invasion and fixation}.
\newblock \emph{\bibinfo{journal}{Physical Review E}}
  \textbf{\bibinfo{volume}{74}}, \bibinfo{pages}{011909}
  (\bibinfo{year}{2006}).

\bibitem{imhof:JMB:2006}
\bibinfo{author}{Imhof, L.~A.} \& \bibinfo{author}{Nowak, M.~A.}
\newblock \bibinfo{title}{Evolutionary game dynamics in a {W}right-{F}isher
  process}.
\newblock \emph{\bibinfo{journal}{Journal of Mathematical Biology}}
  \textbf{\bibinfo{volume}{52}}, \bibinfo{pages}{667--681}
  (\bibinfo{year}{2006}).

\bibitem{imhof:PRSB:2010}
\bibinfo{author}{Imhof, L.~A.} \& \bibinfo{author}{Nowak, M.~A.}
\newblock \bibinfo{title}{Stochastic evolutionary dynamics of direct
  reciprocity}.
\newblock \emph{\bibinfo{journal}{Proceedings of the Royal Society B}}
  \textbf{\bibinfo{volume}{277}}, \bibinfo{pages}{463--468}
  (\bibinfo{year}{2010}).

\bibitem{tainaka:PLA:1993}
\bibinfo{author}{Tainaka, K.}
\newblock \bibinfo{title}{Paradoxial effect in a three-candidate voter model}.
\newblock \emph{\bibinfo{journal}{Physics Letters A}}
  \textbf{\bibinfo{volume}{176}}, \bibinfo{pages}{303--306}
  (\bibinfo{year}{1993}).

\bibitem{frean:PRSB:2001}
\bibinfo{author}{Frean, M.} \& \bibinfo{author}{Abraham, E.~R.}
\newblock \bibinfo{title}{Rock-scissors-paper and the survival of the weakest}.
\newblock \emph{\bibinfo{journal}{Proceedings of the Royal Society B}}
  \textbf{\bibinfo{volume}{268}}, \bibinfo{pages}{1323--1327}
  (\bibinfo{year}{2001}).

\bibitem{Akin:chapter:2016}
\bibinfo{author}{Akin, E.}
\newblock \bibinfo{title}{The iterated prisoner's dilemma: {G}ood strategies
  and their dynamics}.
\newblock In \bibinfo{editor}{Assani, I.} (ed.)
  \emph{\bibinfo{booktitle}{Ergodic Theory, Advances in Dynamics}},
  \bibinfo{pages}{77--107} (\bibinfo{publisher}{de Gruyter},
  \bibinfo{address}{Berlin}, \bibinfo{year}{2016}).

\bibitem{nowak:AMC:1989}
\bibinfo{author}{Nowak, M.~A.} \& \bibinfo{author}{Sigmund, K.}
\newblock \bibinfo{title}{Game-dyamical aspects of the prisoner's dilemma}.
\newblock \emph{\bibinfo{journal}{Applied Mathematics and Computation}}
  \textbf{\bibinfo{volume}{30}}, \bibinfo{pages}{191--213}
  (\bibinfo{year}{1989}).

\bibitem{brandt:JTB:2006}
\bibinfo{author}{Brandt, H.} \& \bibinfo{author}{Sigmund, K.}
\newblock \bibinfo{title}{The good, the bad and the discriminator - errors in
  direct and indirect reciprocity}.
\newblock \emph{\bibinfo{journal}{Journal of Theoretical Biology}}
  \textbf{\bibinfo{volume}{239}}, \bibinfo{pages}{183--194}
  (\bibinfo{year}{2006}).

\bibitem{hilbe:BMB:2011}
\bibinfo{author}{Hilbe, C.}
\newblock \bibinfo{title}{Local replicator dynamics: A simple link between
  deterministic and stochastic models of evolutionary game theory}.
\newblock \emph{\bibinfo{journal}{Bulletin of Mathematical Biology}}
  \textbf{\bibinfo{volume}{73}}, \bibinfo{pages}{2068--2087}
  (\bibinfo{year}{2011}).

\bibitem{grujic:jtb:2012}
\bibinfo{author}{Grujic, J.}, \bibinfo{author}{Cuesta, J.~A.} \&
  \bibinfo{author}{Sanchez, A.}
\newblock \bibinfo{title}{On the coexistence of cooperators, defectors and
  conditional cooperators in the multiplayer iterated prisoner's dilemma}.
\newblock \emph{\bibinfo{journal}{Journal of Theoretical Biology}}
  \textbf{\bibinfo{volume}{300}}, \bibinfo{pages}{299--308}
  (\bibinfo{year}{2012}).

\bibitem{Rodriguez:JMB:2016}
\bibinfo{author}{N\'u\~nez Rodr\'iguez, I.} \& \bibinfo{author}{Neves, A.
  G.~M.}
\newblock \bibinfo{title}{Evolution of cooperation in a particular case of the
  infinitely repeated prisoner's dilemma with three strategies}.
\newblock \emph{\bibinfo{journal}{Journal of Mathematical Biology}}
  \textbf{\bibinfo{volume}{73}}, \bibinfo{pages}{1665--1690}
  (\bibinfo{year}{2016}).

\bibitem{szabo:PRE:2000b}
\bibinfo{author}{Szab{\'o}, G.}, \bibinfo{author}{Antal, T.},
  \bibinfo{author}{Szab{\'o}, P.} \& \bibinfo{author}{Droz, M.}
\newblock \bibinfo{title}{Spatial evolutionary prisoner's dilemma game with
  three strategies and external constraints}.
\newblock \emph{\bibinfo{journal}{Physical Review E}}
  \textbf{\bibinfo{volume}{62}}, \bibinfo{pages}{1095--1103}
  (\bibinfo{year}{2000}).

\bibitem{szolnoki:scirep:2014}
\bibinfo{author}{Szolnoki, A.} \& \bibinfo{author}{Perc, M.}
\newblock \bibinfo{title}{Defection and extortion as unexpected catalysts of
  unconditional cooperation in structured populations}.
\newblock \emph{\bibinfo{journal}{Scientific Reports}}
  \textbf{\bibinfo{volume}{4}}, \bibinfo{pages}{5496} (\bibinfo{year}{2014}).

\bibitem{szolnoki:pre:2014}
\bibinfo{author}{Szolnoki, A.} \& \bibinfo{author}{Perc, M.}
\newblock \bibinfo{title}{Evolution of extortion in structured populations}.
\newblock \emph{\bibinfo{journal}{Physical Review E}}
  \textbf{\bibinfo{volume}{89}}, \bibinfo{pages}{022804}
  (\bibinfo{year}{2014}).

\bibitem{van-veelen:PNAS:2012}
\bibinfo{author}{van Veelen, M.}, \bibinfo{author}{Garc\'ia, J.},
  \bibinfo{author}{Rand, D.~G.} \& \bibinfo{author}{Nowak, M.~A.}
\newblock \bibinfo{title}{Direct reciprocity in structured populations}.
\newblock \emph{\bibinfo{journal}{Proceedings of the National Academy of
  Sciences USA}} \textbf{\bibinfo{volume}{109}}, \bibinfo{pages}{9929--9934}
  (\bibinfo{year}{2012}).

\bibitem{hauert:PRSB:1997}
\bibinfo{author}{Hauert, C.} \& \bibinfo{author}{Schuster, H.~G.}
\newblock \bibinfo{title}{Effects of increasing the number of players and
  memory size in the iterated prisoner's dilemma: a numerical approach}.
\newblock \emph{\bibinfo{journal}{Proceedings of the Royal Society B}}
  \textbf{\bibinfo{volume}{264}}, \bibinfo{pages}{513--519}
  (\bibinfo{year}{1997}).

\bibitem{hauert:JTB:2002b}
\bibinfo{author}{Hauert, C.} \& \bibinfo{author}{Stenull, O.}
\newblock \bibinfo{title}{Simple adaptive strategy wins the prisoner's
  dilemma.}
\newblock \emph{\bibinfo{journal}{Journal of Theoretical Biology}}
  \textbf{\bibinfo{volume}{218}}, \bibinfo{pages}{261--72}
  (\bibinfo{year}{2002}).

\bibitem{stewart:scirep:2016}
\bibinfo{author}{Stewart, A.~J.} \& \bibinfo{author}{Plotkin, J.~B.}
\newblock \bibinfo{title}{Small groups and long memories promote cooperation}.
\newblock \emph{\bibinfo{journal}{Scientific Reports}}
  \textbf{\bibinfo{volume}{6}}, \bibinfo{pages}{26889} (\bibinfo{year}{2016}).

\bibitem{hilbe:pnas:2017}
\bibinfo{author}{Hilbe, C.}, \bibinfo{author}{Martinez-Vaquero, L.~A.},
  \bibinfo{author}{Chatterjee, K.} \& \bibinfo{author}{Nowak, M.~A.}
\newblock \bibinfo{title}{Memory-$n$ strategies of direct reciprocity}.
\newblock \emph{\bibinfo{journal}{Proceedings of the National Academy of
  Sciences USA}} \textbf{\bibinfo{volume}{114}}, \bibinfo{pages}{4715--4720}
  (\bibinfo{year}{2017}).

\bibitem{Li:NatCS:2022}
\bibinfo{author}{Li, J.} \emph{et~al.}
\newblock \bibinfo{title}{Evolution of cooperation through cumulative
  reciprocity}.
\newblock \emph{\bibinfo{journal}{Nature Computational Science}}
  \textbf{\bibinfo{volume}{2}}, \bibinfo{pages}{677--686}
  (\bibinfo{year}{2022}).

\bibitem{taylor:MB:1978}
\bibinfo{author}{Taylor, P.~D.} \& \bibinfo{author}{Jonker, L.}
\newblock \bibinfo{title}{Evolutionarily stable strategies and game dynamics}.
\newblock \emph{\bibinfo{journal}{Mathematical Biosciences}}
  \textbf{\bibinfo{volume}{40}}, \bibinfo{pages}{145--156}
  (\bibinfo{year}{1978}).

\bibitem{geritz:EER:1998}
\bibinfo{author}{Geritz, S. A.~H.}, \bibinfo{author}{Kisdi, E.},
  \bibinfo{author}{Mesz\'{e}na, G.} \& \bibinfo{author}{Metz, J. A.~J.}
\newblock \bibinfo{title}{Evolutionarily singular strategies and the adaptive
  growth and branching of the evolutionary tree}.
\newblock \emph{\bibinfo{journal}{Evolutionary Ecology Research}}
  \textbf{\bibinfo{volume}{12}}, \bibinfo{pages}{35--57}
  (\bibinfo{year}{1998}).

\bibitem{willensdorfer2005mutation}
\bibinfo{author}{Willensdorfer, M.} \& \bibinfo{author}{Nowak, M.~A.}
\newblock \bibinfo{title}{Mutation in evolutionary games can increase average
  fitness at equilibrium}.
\newblock \emph{\bibinfo{journal}{Journal of theoretical biology}}
  \textbf{\bibinfo{volume}{237}}, \bibinfo{pages}{355--362}
  (\bibinfo{year}{2005}).

\bibitem{lorberbaum:JTB:1994}
\bibinfo{author}{Lorberbaum, J. M.~D.}
\newblock \bibinfo{title}{No strategy is evolutionary stable in the repeated
  {P}risoner's {D}ilemma}.
\newblock \emph{\bibinfo{journal}{Journal of Theoretical Biology}}
  \textbf{\bibinfo{volume}{168}}, \bibinfo{pages}{117--130}
  (\bibinfo{year}{1994}).

\bibitem{Lorberbaum:JTB:2002}
\bibinfo{author}{Lorberbaum, J.~P.}, \bibinfo{author}{Bohning, D.~E.},
  \bibinfo{author}{Shastri, A.} \& \bibinfo{author}{Sine, L.~E.}
\newblock \bibinfo{title}{Are there really no evolutionarily stable strategies
  in the iterated prisoner's dilemma?}
\newblock \emph{\bibinfo{journal}{Journal of Theoretical Biology}}
  \textbf{\bibinfo{volume}{214}}, \bibinfo{pages}{155--169}
  (\bibinfo{year}{2002}).

\bibitem{traulsen:PNAS:2009}
\bibinfo{author}{Traulsen, A.}, \bibinfo{author}{Hauert, C.},
  \bibinfo{author}{De~Silva, H.}, \bibinfo{author}{Nowak, M.~A.} \&
  \bibinfo{author}{Sigmund, K.}
\newblock \bibinfo{title}{Exploration dynamics in evolutionary games}.
\newblock \emph{\bibinfo{journal}{Proceedings of the National Academy of
  Sciences USA}} \textbf{\bibinfo{volume}{106}}, \bibinfo{pages}{709--712}
  (\bibinfo{year}{2009}).

\bibitem{Ramirez:Arxiv:2022}
\bibinfo{author}{Ram\'irez, M.~A.}, \bibinfo{author}{Smerlak, M.},
  \bibinfo{author}{Traulsen, A.} \& \bibinfo{author}{Jost, J.}
\newblock \bibinfo{title}{Diversity enables the jump towards cooperation for
  the traveler's dilemma}.
\newblock \emph{\bibinfo{journal}{{arXiv}}}
  \textbf{\bibinfo{volume}{https://arxiv.org/pdf/2210.15971}}
  (\bibinfo{year}{2022}).

\bibitem{mcnamara:Nature:2004}
\bibinfo{author}{McNamara, J.~M.}, \bibinfo{author}{Barta, Z.} \&
  \bibinfo{author}{Houston, A.~I.}
\newblock \bibinfo{title}{Variation in behaviour promotes cooperation in the
  {P}risoner's {D}ilemma game}.
\newblock \emph{\bibinfo{journal}{Nature}} \textbf{\bibinfo{volume}{428}},
  \bibinfo{pages}{745--748} (\bibinfo{year}{2004}).

\bibitem{Boyd:HumEcol:1982}
\bibinfo{author}{Boyd, R.} \& \bibinfo{author}{Richerson, P.~J.}
\newblock \bibinfo{title}{Cultural transmission and the evolution of
  cooperative behavior}.
\newblock \emph{\bibinfo{journal}{Human ecology}}
  \textbf{\bibinfo{volume}{10}}, \bibinfo{pages}{325--351}
  (\bibinfo{year}{1982}).

\bibitem{traulsen:PNAS:2006}
\bibinfo{author}{Traulsen, A.} \& \bibinfo{author}{Nowak, M.~A.}
\newblock \bibinfo{title}{Evolution of cooperation by multi-level selection}.
\newblock \emph{\bibinfo{journal}{Proceedings of the National Academy of
  Sciences USA}} \textbf{\bibinfo{volume}{103}}, \bibinfo{pages}{10952--10955}
  (\bibinfo{year}{2006}).

\bibitem{Hauert:PNAS:2021}
\bibinfo{author}{Hauert, C.} \& \bibinfo{author}{Doebeli, M.}
\newblock \bibinfo{title}{Spatial social dilemmas promote diversity}.
\newblock \emph{\bibinfo{journal}{Proceedings of the National Academy of
  Sciences}} \textbf{\bibinfo{volume}{118}}, \bibinfo{pages}{e2105252118}
  (\bibinfo{year}{2021}).

\end{thebibliography}

\end{document}